\documentclass[12pt]{article}
\usepackage[margin=1in]{geometry}

\usepackage{algorithm}
\usepackage{algorithmic}
\usepackage{booktabs} 
\usepackage{amsthm}
\usepackage{natbib}
\usepackage{amsmath}
\usepackage{mathtools}
\usepackage{amssymb}
\usepackage[utf8]{inputenc}
\usepackage{graphicx}
\usepackage{xcolor}
\usepackage{comment}
\usepackage[english]{babel}
\usepackage{bm}
\usepackage{multirow}
\usepackage{url}

\usepackage{xcolor}
\usepackage{placeins}
\usepackage{multirow}
\usepackage{hyperref}

\usepackage{graphicx} 
\usepackage{amsthm}
\usepackage{amsmath}
\usepackage{amssymb}  
\usepackage[margin=1in]{geometry}
\usepackage{xcolor}
\usepackage{cleveref}
\usepackage{verbatim}
\usepackage{bm}
\usepackage{booktabs}
\usepackage{graphicx} 
\usepackage{makecell} 

\usepackage{algorithm}
\usepackage{algorithmic}
\usepackage{subcaption}
\usepackage{authblk}

\newcommand\R{{\mathbb R}}

\newcommand\X{{\mathbf X}}

\newcommand\Xn{\X_n}

\newcommand\Pn{\mathbb P^{(n)}}
\newcommand\pn{p^{(n)}}
\newcommand\mN{{\mathcal{N}}}
\newcommand\E{{\mathbb E}}

\newcommand\bZ{{\mathbf Z}}
\newcommand\bI{{\mathbf I}}

\newcommand\mD{{\mathcal D}}
\newcommand\mX{{\mathcal X}}

\newcommand\mG{{\mathcal G}}

\newcommand\wt{\widetilde}
\newcommand\wh{\widehat}

\newcommand\diff{\mathrm d}

\newcommand{\norm}[1]{\left\lVert#1\right\rVert}

\newcommand{\normbig}[1]{\big\lVert#1\big\rVert}

\newcommand\iid{\overset{\text{iid}}{\sim}}

\newcommand\MLE{\widehat{\theta}^{\text{\rm MLE}}_n}
\newcommand\zero{\widehat{\theta}_n}
\newcommand\boot{\widehat{\theta}^{(b)}}
\newcommand\sh{\widehat{s}}
\newcommand\Ih{\widehat{I}}

\newtheorem{theorem}{Theorem}
\newtheorem{lemma}{Lemma}

\newtheorem{assumption}{Assumption}
\newtheorem{remark}{Remark}
\newtheorem{example}{Example}

\begin{document}

\onecolumn

\title{\LARGE\bf Likelihood-Free Inference via Structured Score Matching}

\author[1]{Haoyu Jiang}
\author[1]{Yuexi Wang}
\author[2]{Yun Yang}

\affil[1]{{Department of Statistics}, {University of Illinois Urbana-Champaign}, {USA}}
\affil[2]{{Department of Mathematics}, {University of Maryland, College Park}, 
{USA}}

\maketitle

\begin{abstract}
In many statistical problems, the data distribution is specified through a generative process for which the likelihood function is analytically intractable, yet inference on the associated model parameters remains of primary interest. We develop a likelihood-free inference framework that combines score matching with gradient-based optimization and bootstrap procedures to facilitate parameter estimation together with uncertainty quantification. The proposed methodology introduces tailored score-matching estimators for approximating likelihood score functions, and incorporates an architectural regularization scheme that embeds the statistical structure of log-likelihood scores to improve both accuracy and scalability. We provide theoretical guarantees and demonstrate the practical utility of the method through numerical experiments, where it performs favorably compared to existing approaches.
\end{abstract}

\section{Introduction}
In many scientific domains, including astrophysics \citep{mishra2022neural}, epidemiology \citep{ionides2015inference,chatha2024neural}, ecology \citep{wood2010statistical,beaumont2010approximate}, and genetics \citep{beaumont2002approximate}, researchers model complex phenomena using forward simulators. These simulators implicitly define a likelihood function but rarely provide it in closed form, rendering traditional statistical methods that rely on explicit likelihood evaluation inapplicable. This framework is commonly referred to as \emph{likelihood-free inference} (LFI) or \emph{simulation-based inference} (SBI) \citep{cranmer2020frontier}.

A large body of work in LFI has focused on the Bayesian setting, where a prior over the parameter space is specified and the goal is to approximate the posterior distribution. The most established approach is Approximate Bayesian Computation \citep[ABC,][]{beaumont2002approximate}, which assigns weights to parameter draws based on the similarity between simulated and observed datasets. Much of the research in this area has focused on developing more effective similarity measures, including distances on summary statistics \citep{beaumont2002approximate,fearnhead2011constructing, pacchiardi2022score} 
and discrepancies between empirical distributions \citep{jiang2018approximate,bernton2019approximate,park2016k2,wang2022approximate,cherief2020mmd}. Another major direction is Bayesian synthetic likelihood \citep{price2018bayesian,frazier2025synthetic}, which constructs a surrogate likelihood based on summary statistics that are asymptotically normally distributed. More recently, advances in deep learning have enabled direct approximation of the posterior distribution via generative models, such as autoregressive flows~\citep{papamakarios2016fast,papamakarios2017masked,papamakarios2019sequential}, conditional generative adversarial networks \citep{wang2022adversarial}, and conditional diffusion models \citep{sharrock2024sequential}.
However, a key limitation of these approaches is that their performance depends heavily on the choice of training objectives and simulation budgets. As a result, the learned posteriors may be biased or poorly calibrated, particularly in regions of the parameter space with limited training coverage~\citep{frazier2018asymptotic,frazier2024statistical}.

In the frequentist setting, LFI is often referred to as ``indirect inference'' \citep{gourieroux1993indirect,jiang2004indirect} in the econometrics literature, with a focus on parameter estimation, confidence sets and hypothesis testing \citep{dalmasso2020confidence,dalmasso2024likelihood}. One major line of work develops estimators by matching simulated and observed datasets \citep{mcfadden1989method,gallant1996moments,frazier2019indirect,kaji2023adversarial}, where the discrepancy functions play a role analogous to the similarity measures in ABC. A second line introduces a tractable auxiliary model and compares simulated and observed datasets through auxiliary estimates, with parameters chosen to minimize their difference \citep{gourieroux1993indirect}. The more recent direction, similar to the Bayesian setting, utilizes deep neural network or deep generative models to directly approximate the odds ratio  \citep{dalmasso2024likelihood} or the Fisher score \citep{khoo2025direct} from simulated datasets.

In this work, we adopt the frequentist perspective on likelihood-free inference and propose a new score-based estimator for the model parameter, defined as the root of an estimated likelihood score obtained via score matching, together with associated confidence sets. Our approach is motivated by the fact that the MLE can be characterized as the root of the likelihood score.  To improve the quality of the score approximation, we use structured score matching, which incorporates key statistical properties of the true likelihood score --- such as additivity across samples, the Fisher information identity, and the mean-zero property --- directly into the network architecture and training objective (c.f.~Section~\ref{se:score}). This ensures that the estimated score respects the fundamental structure of the likelihood score and leads to more accurate inference. In addition, we develop a computation procedure for this estimator based on an iterative algorithm that mimics the gradient dynamics of maximizing the log-likelihood function. For uncertainty quantification, we consider three options for constructing confidence sets: using the estimated Fisher information matrix, the Huber sandwich covariance matrix, and multiplier bootstrap methods. We provide theoretical guarantees on the consistency and asymptotic behavior of our estimators, establish bootstrap consistency, and show that the algorithm converges locally at an exponential rate. Empirically, our methods demonstrate strong performance compared to existing frequentist LFI approaches.

\section{Problem Setup and Methods}
Let $\Xn^\ast = (X_{1}^\ast, \ldots, X_{n}^\ast)$ denote a collection of $n$ observations drawn from a distribution $\Pn_{\theta^\ast}$ in the parametric model family $\big\{\Pn_{\theta} : \theta \in \Theta \big\}$, 
where each $X^\ast_{i} \in \R^p$ and $\theta^\ast$ denotes the true parameter lying in the parameter space $\Theta \subset \R^{d_\theta}$. For simplicity, we assume that for every $\theta \in \Theta$, the distribution $\Pn_\theta$ admits a density $\pn_\theta$.  In this work, we focus on the i.i.d.~setting where $\pn_\theta(\Xn) = \prod_{i=1}^n p_\theta(X_i)$. 

In the likelihood-free inference framework, the likelihood function $\log \pn_\theta(\Xn) = \sum_{i=1}^n \log p_\theta(X_i)$ may not be directly available. Instead, $\Pn_\theta$ is implicitly specified through a data-generating process $\mG_\theta$ parameterized by $\theta$, in the following sense: for any $\theta \in \Theta$, if one draws a random vector $\bZ = (Z_1, \ldots, Z_n)$ with i.i.d.~components $Z_i \sim \mu$ for some known distribution $\mu$ (e.g., $\mN(\bf{0},\, \bI_d)$), then $\big(\mG_\theta(Z_1), \ldots, \mG_\theta(Z_n)\big) \sim \Pn_\theta$. In other words, while the likelihood cannot be directly evaluated, simulation is feasible: for any parameter $\theta \in \Theta$, one can readily use the simulator $\mG_\theta$ to generate pseudo-datasets $ {\bf X}_m \sim \mathbb P^{(m)}_\theta$ of arbitrary size $m$.

\begin{example}[A toy example]\label{ex:toy} As a simple illustration, consider data generated as $X = \exp(Z_1) + Z_2$, where
\begin{align*}
% , \quad\mbox{where}\\
\begin{pmatrix}
Z_1\\
Z_2 
\end{pmatrix} 
&\sim \mathcal{N}\left(
\begin{pmatrix}
\theta_1\\
\theta_2 
\end{pmatrix}
, \begin{pmatrix}
1 & 0.2\\
0.2 & 1 
\end{pmatrix}\right).
\end{align*}
Here $\theta =(\theta_1, \theta_2)$ is the parameter of interest. Despite the simplicity of this model, its likelihood $p(X |\,\theta)$ is analytically intractable, even though the model is regular in the sense that it satisfies classical asymptotic properties such as consistency, asymptotic normality, and efficiency of the maximum likelihood estimator \citep{Vaart_1998}. 
\end{example}

Other problems involving intractable likelihoods, such as complex or highly nonlinear latent variable models~\citep{martin2014approximate, pavone2025phylogenetic}, can likewise be addressed using the likelihood-free inference framework, as they permit efficient sampling despite the difficulty of direct likelihood evaluation.

One line of work in the likelihood-free inference literature proposes replacing the intractable likelihood with a surrogate objective that facilitates parameter estimation~\citep[e.g.,][]{gutmann2016bayesian,beaumont2002approximate}. While these surrogate-based methods offer computational advantages, they may lead to a loss in statistical efficiency. In this work, we aim to directly approximate the maximum likelihood estimator $\MLE$, defined as
\[
\MLE = \arg\max_{\theta \in \mathbb{R}^{d_\theta}} \sum_{i=1}^n \log p_\theta(X_i).
\]
The MLE is known to be asymptotically efficient, in the sense that its asymptotic variance attains the Cram\'{e}r--Rao lower bound~\citep{Vaart_1998}. Moreover, to enable uncertainty quantification, we aim to construct a confidence set $R(\mathcal{X}_n^*)$ (centered around $\widehat{\theta}_{\mathrm{MLE}}$) for $\theta^*$ with asymptotic $(1 - \alpha)$ nominal coverage, that is,
\[
\mathbb{P}_{\theta^*}^{(n)}\big(\theta^* \in R(\Xn^*)\big) \to 1 - \alpha, \quad \text{as } n \to \infty.
\]

Before delving into the technical details, we first introduce some notation. We use the shorthand $s^\ast(\theta, \Xn) = \nabla_\theta \log \pn_\theta(\Xn)$ to denote the true likelihood score.
We denote the Fisher information matrix at parameter value $\theta$ by $I(\theta) = \E_{X \sim p_\theta}\big[-\nabla_\theta s^\ast(\theta, X)\big] = \E_{X \sim p_\theta}\big[-\nabla_\theta^2 \log p_\theta(X)\big]$.

In the remainder of this section, we first introduce a structured score matching procedure, which tailors the generic score matching approach~\citep{hyvarinen2007some} to the estimation of the likelihood score, and then describe how we construct the score-based estimator and confidence sets based on it.

\subsection{Structured score matching}\label{se:score}
Our estimator relies on the likelihood score. Since a likelihood-free framework does not assume the availability of a tractable likelihood, we employ the score matching technique proposed by \citet{hyvarinen2007some}, \citet{song2019generative} and~\citet{meng2020autoregressive} to approximate the likelihood score $\nabla_\theta \log \pn_\theta(\Xn^\ast)$ with a score model $s_\phi(\theta, \Xn): \R^{d_\theta} \times \R^{np} \to \R^{d_\theta}$, parametrized by some parameter $\phi$ (e.g., a neural network). When the sample size $n$ is large, learning or approximating a generic nonlinear function with input dimension $d_\theta + np$ may require an enormous number of training samples and parameters in $\phi$. This motivates us to later impose additional structure to simplify the score network architecture.

Specifically, to estimate the likelihood score, we minimize the expected squared distance
between the true score $s^\ast(\theta,\Xn)$ and its estimator $s_\phi(\theta,\Xn)$. This yields the following optimization procedure:
\begin{equation*}
    \min_{\phi} \, 
    \E_{(\theta, \Xn) \sim p(\theta)\pn_\theta(\Xn)} 
    \Big[\big\| s_\phi(\theta, \Xn) -  s^*(\theta, \Xn) \big\|^2\Big],
\end{equation*}
where $p(\theta)$ is a sampling distribution that controls the weighting of the score estimation loss across the parameter space $\Theta$ .
The sampling distribution $p(\theta)$ is assumed to have full support over $\Theta$. For example, one may choose it as a uniform distribution over an expert-specified range or as a Gaussian distribution centered at a rough estimate $\theta^{(0)}$ \citep{khoo2025direct}.

\begin{remark}[A two-round procedure]\label{rem:two} The quality of our estimators depends on how well $ s_{\widehat{\phi}}(\theta,\Xn)$ approximates the true score $s^*(\theta,\Xn)$ in a neighborhood of $\theta^\ast$ (see \Cref{ass:uniform_sm_err_single}). In high-dimensional settings where $p$ is large, however, the weighting distribution $p(\theta)$ may place little mass in this neighborhood, yielding a poor estimate of $s_{\widehat{\phi}}(\theta,\Xn^*)$. To address this, one may adopt a two-round procedure: in the first round, simulations are drawn from the uninformative proposal $p(\theta)$ to obtain an initial estimator and its confidence set; in the second round, the sampling distribution is refined to concentrate around the initial estimate, and $\wh\theta_n$ is re-estimated, leading to improved score approximation and more accurate inference. We implement the two-round procedure in our empirical analysis in \Cref{sec:simulation} and include a comparison with the estimators from the first round in Appendix.
\end{remark}

Under mild conditions (see \citet{hyvarinen2005estimation,jiang2025simulation} for details), the score matching optimization problem is equivalent to  
{
\begin{equation}\label{eq:sm_obj}
\begin{split}
&\min_\phi \;  \E_{ (\theta, \Xn) \sim p(\theta)\pn_\theta(\Xn)} \Big[ 
    \norm{s_\phi(\theta, \Xn)}^2  + 2 s_\phi(\theta, \Xn)^\top \nabla_\theta \log p(\theta) + 2\, \text{tr}\!\big(\nabla_\theta s_{\phi}(\theta, \Xn)\big) \Big],
\end{split}
\end{equation}}
 where $\text{tr}(A)$ denotes the matrix trace. We include a proof in the Appendix, adapted from \citep{hyvarinen2005estimation,jiang2025simulation}, to keep our presentation self-contained.
This equivalent formulation enables score matching without explicitly computing $s^\ast$, by instead solving an empirical risk minimization problem over $\phi$. 

However, as previously noted, the vanilla score matching strategy may require large training samples and may fail to capture the statistical structure of the true score, which is essential for our estimators (for reasons clarified later). To address this, we adopt the structured score matching approach, where a similar setting under the Bayesian framework was considered in \citet{jiang2025simulation}. We incorporate three universal properties of the likelihood score into both the network architecture and the training objective.

Here we briefly review three key statistical properties of $s^*(\theta, \Xn)$ and how they are incorporated into $s_{\phi}(\theta, \Xn)$. From classical statistical theory, under suitable regularity conditions, the true likelihood score satisfies the following properties \citep[see, e.g., Chapter~5.5 of][]{Vaart_1998}:
1.~\emph{(additive structure)} 
   $s^\ast(\theta, \Xn) = \sum_{i=1}^n s^\ast(\theta, X_i)$.  
2.~\emph{(curvature structure)} 
   $\E_{X \sim p_\theta}\!\left[s^\ast(\theta, X) s^\ast(\theta, X)^\top 
   + \nabla_\theta s^\ast(\theta, X)\right] = 0$.  
3.~\emph{(mean-zero structure)} 
   $\E_{X \sim p_\theta}[s^\ast(\theta, X)] = 0$.  
Accordingly, the estimated score function $s_{\widehat\phi}(\theta, \Xn)$ should inherit these properties.  In particular, the curvature property implies two equivalent definitions of the Fisher information matrix:  
$I(\theta) = \E_{X \sim p_\theta}\big[-\nabla_\theta s^\ast(\theta, X)\big] 
= \E_{X \sim p_\theta}\big[s^\ast(\theta, X)\,s^\ast(\theta, X)^\top\big]$.
For the additive structure, we impose $s_\phi(\theta, \Xn) = \sum_{i=1}^n s_\phi(\theta, X_i)$.  
For the curvature structure, we require $\E_{P_\theta}[s_{\widehat\phi}(\theta, X)s_{\widehat\phi}(\theta, X)^\top + \nabla_\theta s_{\widehat\phi}(\theta, X)] \approx 0$, which is enforced through a regularization penalty.  
For the mean-zero structure, $\E_{P_\theta}[s_{\widehat\phi}(\theta, X)] \approx 0$, which is achieved via a second-step debiasing regression.  
 We also include the detailed description of the algorithm in Appendix. For notational simplicity, we henceforth denote the estimated score function as $\sh(\theta, \Xn) = s_{\widehat\phi}(\theta,\Xn)$.

\begin{remark}[Choice of score matching scheme]
While denoising score matching \citep{vincent2011connection} is often preferred in diffusion models, we instead adopt direct score matching in \eqref{eq:sm_obj}. This choice allows a more transparent incorporation of the underlying statistical structure and is well suited to SBI settings, which typically involve low- to moderate-dimensional parameters. In particular, the resulting additive structure reduces network complexity, while the mean-zero and curvature conditions enable precise analysis and control of error accumulation over $n$ observations (see \Cref{sec:theory}). In contrast, although denoising score matching avoids evaluating the Jacobian trace $\nabla_\theta s_\phi(\theta, {\mathbf{X}_n})$, the Gaussian convolution step destroys the conditional i.i.d. structure of $\Xn$ given the perturbed parameter $\tilde{\theta}$. As a result, the statistical structure exploited in our analysis no longer applies. Recent work \citep{geffner2023compositional,arruda2025compositional} incorporates additive structures within denoising score matching, but the nonzero Gaussian perturbation introduces bias, and how this bias accumulates with $n$ has not yet been fully analyzed. 
To mitigate the computational cost associated with Jacobian evaluation, one may also consider sliced score matching variants \citep{song2020sliced,pacchiardi2022score}. 
\end{remark}

\subsection{Score-based estimator}\label{sec:sco-based-est}

Under mild conditions, the likelihood score function $s^\ast(\theta,\,\Xn)=\sum_{i=1}^n s^\ast(\theta, X_i^\ast)$, viewed as a function of the parameter $\theta$, admits a unique solution \citep[see, e.g., Section~5.6 of][]{Vaart_1998} within a neighborhood $\mathcal B(\theta^\ast; r_0) = \big\{\theta \in \mathbb R^{d_\theta}:\,\|\theta-\theta^\ast\|\le r_0\big\}$ of $\theta^\ast$, where $r_0>0$ is a fixed radius such that the Fisher information matrix $I(\theta)$ is locally positive definite. The MLE $\MLE$ can be characterized as the root of the score function in this neighborhood, i.e.,
\begin{equation}\label{eq:MLE_root}
\MLE = \Big\{ \theta \in \mathcal B(\theta^\ast;r_0): \sum_{i=1}^n s^\ast(\theta, X_i^\ast)=0 \Big\}.
\end{equation}
Motivated by this, we define our score-based estimator $\zero$ for approximating $\MLE$ as the root of the estimated score $\sh$, that is,
\begin{equation}\label{eq:shat_root}
\zero = \Big\{ \theta \in \mathcal B(\theta^\ast;r_0): \sum_{i=1}^n \sh(\theta, X_i^\ast)=0 \Big\}.
\end{equation}
Our theoretical results in Section~\ref{sec:theory} establish that the estimated score $\sh$ admits at least one root in $\mathcal B(\theta^\ast; r_0)$, and moreover that this root is unique under suitable assumptions on the quality of the estimated score, which ensures the estimator is well-defined. Showing the existence of such a root is not trivial  (see the discussion after Theorem~\ref{thm:consistency}), since the estimated score $\sh$ may not be the gradient of a function; therefore, one cannot simply invoke the existence of a critical point (or more precisely, a local minimum) of a continuous function.

In practical implementation, we propose to numerically compute $\zero$ as the limit of the following iterative scheme:
\begin{align}\label{eq:GA}
    \theta^{(t+1)} = \theta^{(t)} + \alpha H_t \,\sh(\theta^{(t)},\,\Xn^\ast), \quad t \geq 0,
\end{align}
where $\alpha$ is the step size and $H_t$ is a preconditioning matrix. In practice, one may either use a random initialization or apply a computationally inexpensive method such as the simulated method of moments \citep{mcfadden1989method,pakes1989simulation} to set $\theta^{(0)}$. Our numerical experiments suggest that random initialization alone is already sufficient.

Note that procedure~\eqref{eq:GA} coincides with a preconditioned gradient ascent algorithm for maximizing the log-likelihood $\log p_\theta^{(n)}(\Xn^\ast)$ when $\sh$ is replaced by the true likelihood score $s^\ast$. 
Moreover, any stationary point of this recursion is a root of $\sh(\,\cdot\,,\,\Xn^\ast)$. In Section~\ref{sec:theory}, we further show that Algorithm~\eqref{eq:GA} converges exponentially fast to $\zero$ when initialized within the neighborhood $\mathcal B(\theta^\ast; r_0)$. This is because our score-matching method enforces curvature matching, ensuring that the Jacobian $\nabla_\theta \sh(\theta^{(t)},\,\Xn^\ast)$ remains close to the true likelihood-score Jacobian $\nabla_\theta s^\ast(\theta^{(t)},\,\Xn^\ast)$ (i.e., the Hessian of $\log p_\theta^{(n)}(\Xn^\ast)$). As a result, the iterative scheme~\eqref{eq:GA} inherits the (local) strong convexity and smoothness properties and thus effectively mimics the gradient dynamics of maximizing the log-likelihood function (c.f.~Theorem~\ref{thm:GA_conv}).

Regarding the preconditioning matrix $H_t$, a simple choice is to set $H_t \equiv \bm{I}_{d_\theta}$, which corresponds to the vanilla gradient ascent algorithm. However, thanks to the curvature matching property of our score-matching method, a more algorithmically efficient choice is $H_t = \big[\nabla_\theta \sh(\theta^{(t)},\,\Xn^\ast)\big]^{-1}$, which corresponds to the Newton–Raphson algorithm  \citep[see, e.g.,][]{bickel2015mathematical} for maximizing the log-likelihood, whose exact counterpart enjoys super-exponential convergence. In particular, since we approximate the score function $\sh(\theta, \Xn)$ with a neural network, the derivative of the score network can be efficiently computed via automatic differentiation \citep{baydin2018autodiff}.

When the parameter dimension $d_\theta$ is large, repeatedly computing a full matrix inverse at each iteration may be computationally expensive. In such cases, one may instead set $H_t \equiv \big[\nabla_\theta \sh(\theta^{(0)},\,\Xn^\ast)\big]^{-1}$, provided a good initialization $\theta^{(0)}$ is available (possibly obtained via a short pilot run). This corresponds to a quasi-Newton–type algorithm \citep[see, e.g.,][]{nocedal2006numerical} for maximizing the log-likelihood. In addition, classical statistical theory on one-step MLE \citep[see, e.g., Section~5.7 of][]{Vaart_1998} shows that once $\theta^{(t)}$ enters an $o_p(1)$ neighborhood of $\theta^\ast$, a single Newton–Raphson update is sufficient to obtain an estimator that is asymptotically equivalent to the MLE. In our case, although we only have access to an approximate score curvature $\nabla_\theta \sh$, the curvature-matching penalty ensures that $\nabla_\theta \sh$ remains close to the true curvature $\nabla_\theta s^\ast$ of the log-likelihood, so only a few iterations are sufficient to obtain a statistically optimal estimator. See Fig.~\ref{fig:quasi-newton_cutoff} in Section~\ref{sec:simulation} for a numerical demonstration.

\subsection{Confidence interval construction}
According to Theorem~\ref{thm:asymp_plugin} in Section~\ref{sec:theory}, the score-based estimator $\zero$ is asymptotically indistinguishable from the MLE $\MLE$. From classical asymptotic normality of the MLE, we know 
\[
\sqrt n (\MLE -\theta^*) \to \mN(0, [I(\theta^*)]^{-1}),\ \ n\to\infty,
\]
where recall that $I(\theta^*)=\E_{X\sim p_{\theta^\ast}}[-\nabla_\theta s^*(\theta^\ast, X)]$ is the Fisher information matrix evaluated at $\theta^\ast$.
Motivated by this result, we present three approaches for constructing confidence intervals based on the estimated score $\sh(\theta, \Xn)$. The first two rely on plug-in methods that replace $I(\theta^\ast)$ with a consistent estimator, while the third is based on the (multiplier) bootstrap.

\paragraph{Option 1: Information matrix.} The first approach is based on directly approximating the information matrix $I(\theta^\ast)$ by taking an empirical average of the negative Jacobian of the estimated score $\sh$. Specifically, we consider the following simple plug-in estimator:
\begin{equation}\label{eq:I_hat}
\widehat I (\zero) 
= -\frac{1}{n}\sum_{i=1}^n 
\frac{1}{2}\Big(\nabla_\theta \sh(\theta, X_i^\ast) 
+ \nabla_\theta \sh(\theta, X_i^\ast)^\top\Big)\Big|_{\theta=\zero},
\end{equation}
where the average of the Jacobian and its transpose is taken to ensure that $\widehat I (\zero)$ is symmetric.

The corresponding confidence set $R(\Xn^*)$ for $\theta^\ast$ with significance level $1-\alpha\in(0,1)$ is then constructed as
\begin{equation*}
R(\Xn^*) =\big\{\theta: n(\theta-\zero)^T \widehat I(\zero)(\theta-\zero)\leq \chi^2_{d_\theta, 1-\alpha}\big\}.
\end{equation*}
where $\chi^2_{d_\theta, 1-\alpha}$ is the $(1-\alpha)$-quantile of the chi-square distribution with $d_\theta$ degree of freedom. In addition, marginal confidence intervals for each coordinate $\theta_j^\ast$ with level $1-\alpha\in(0,1)$ can be constructed as 
\begin{equation*}
R_j(\Xn^*) = \Big\{\, \theta_j: \big|\theta_j-[\zero]_j \big|\leq  z_{1-\alpha/2} \sqrt{\big[[\widehat I(\zero)]^{-1}\big]_{jj}} \,\Big\},
\end{equation*}
where $[\bm{a}]_j$ denotes the $j$-th coordinate of a generic vector $\bm{a}$, $[A]_{ij}$ denotes the $(i,j)$-th entry of a generic matrix $A$, and $z_{1-\alpha}$ is the $(1-\alpha)$-quantile of the standard normal distribution.

Another plug-in estimate for the Fisher information matrix, based on the curvature property, can be specified as
\begin{equation}\label{eq:I_ss}
\wh I(\zero) = \frac{1}{n}\sum_{i=1}^n \Big(\wh s(\zero, X_i^*)\, \wh s(\zero, X_i^*)^T\Big)
\end{equation}
Our empirical results in \Cref{sec:simulation} show that the \eqref{eq:I_hat} consistently outperform  \eqref{eq:I_ss}.

\paragraph{Option 2: Sandwich covariance matrix.} Because our estimator is defined as the root of the estimated score rather than that of the true score, it may be viewed as an $M$-estimator under a potentially misspecified model, since the curvature constraint $\mathbb E_{X\sim p_\theta}\!\big[s^\ast(\theta,X) s^\ast(\theta,X)^\top + \nabla_\theta s^\ast(\theta,X)\big]=0$ that holds for the true score $s^\ast$ may not be satisfied by $\sh$. In this case, a more robust and accurate alternative is the Huber sandwich covariance estimator \citep{huber1967behavior}, defined as
\begin{equation*}
\widehat \Sigma_n = \big(\widehat I (\zero)\big)^{-1} 
\Bigg[\frac{1}{n}\sum_{i=1}^n \sh (\zero, X_i^\ast)\,\sh (\zero, X_i^\ast)^\top\Bigg] 
\big(\widehat I (\zero)\big)^{-1}
\end{equation*}
with $\widehat I(\zero)$ being given in \eqref{eq:I_hat}. Note that when a sample version of the curvature constraint holds for $\sh$, then $\widehat  \Sigma_n$ reduces to $\big(\widehat I (\zero)\big)^{-1}$.

Analogous to Option 1, we construct the confidence set for $\theta^\ast$ with level $1-\alpha\in(0,1)$ using the sandwich covariance matrix as
\begin{equation*}
R(\Xn^\ast) = \Big\{\theta : 
n(\theta-\zero)^\top \widehat \Sigma_n^{-1} (\theta-\zero) 
\leq \chi^2_{d_\theta, 1-\alpha}\Big\},
\end{equation*}
together with the corresponding marginal confidence intervals.

The sandwich covariance matrix formula was also considered by \citet{frazier2025synthetic} in the context of Bayesian synthetic likelihood. We provide a detailed discussion about our differences in the Appendix.

\iffalse
\begin{remark}[Comparison with \citet{frazier2025synthetic}]
The sandwich covariance matrix formula was also adopted by \citet{frazier2025synthetic} in the context of Bayesian synthetic likelihood, where they propose to learn a (partial) likelihood function defined through summary statistics, with the score and Hessian derived under the assumption that the summary statistics converge to a normal distribution as $n \to \infty$. In contrast, our approach directly approximates the true likelihood score without discarding any information from the data, while the curvature-matching penalty ensures that the local geometry remains aligned with the true log-likelihood Hessian. Consequently, our estimator is asymptotically efficient and yields asymptotically shortest confidence intervals (see also our empirical results in Section~\ref{sec:simulation}).
\end{remark}
\fi

\paragraph{Option 3: Bootstrap.}
Our final option for uncertainty quantification is the bootstrap \citep{efron1979bootstrap}. 
In particular, we consider the multiplier bootstrap \citep[see, e.g., Section~2.9 of][]{vaart1996weak}, where for each bootstrap replicate $b=1,2,\ldots,B$, we obtain $\zero^{(b)}$ as the root of the weighted estimated score equation
\begin{equation}\label{eq:boot_est}
\sum_{i=1}^n \omega_i^{(b)} \,\sh(\zero^{(b)}, X_i^*) = 0,
\end{equation}
with $\{\omega_i^{(b)}\}_{i=1}^n$ i.i.d.~random multipliers satisfying $\E[\omega_i^{(b)}]=1$ and ${\rm Var}(\omega_i^{(b)})=1$. 
A common choice is $\omega_i^{(b)} \iid \text{Exp}(1)$. Following the same argument as earlier, one can show that the bootstrap estimator $\zero^{(b)}$ is well defined (exists and is locally unique). 

Under mild conditions, the conditional distribution of 
$\sqrt{n}\,(\zero^{(b)}-\zero)$ given $\Xn^\ast$ (with randomness arising from the multipliers) approximates the distribution of $\sqrt{n}\,(\zero-\theta^\ast)$ (cf.~Theorem~\ref{thm:asymp_boot}). 
Consequently, a confidence set for $\theta^\ast$ at significance level $1-\alpha \in (0,1)$ can be constructed as
\begin{equation*}
R(\Xn^*) = \Big\{\theta : (\theta-\zero)^\top H (\theta-\zero)\leq \widehat q_{H,1-\alpha}\Big\},
\end{equation*}
for any positive definite matrix $H$, where $\widehat q_{H,1-\alpha}$ denotes the empirical $(1-\alpha)$-th quantile of the bootstrap quantities
$\big\{(\theta^{(b)}-\zero)^\top H (\theta^{(b)}-\zero)\big\}_{b=1}^B$.
Furthermore, we can directly construct marginal confidence intervals for each coordinate $\theta_j$ as
\[
R_j(\Xn^*) = \Big[[\zero]_j + \widehat a_n,\; [\zero]_j + \widehat b_n\Big],
\]
where $\widehat a_n$ and $\widehat b_n$ denote the empirical $(\alpha/2)$-th and $(1-\alpha/2)$-th quantiles of 
$\big\{[\zero^{(b)}]_j - [\zero]_j\big\}_{b=1}^B$, respectively.

\begin{remark}[Comparison of three options]
Our numerical results in Section~\ref{sec:simulation} suggest that the sandwich covariance plug-in method and the bootstrap method perform better than the information matrix plug-in method, with the former two showing very similar performance. This is consistent with our expectation that the sandwich method provides a more robust and accurate approximation to the asymptotic covariance matrix, while under model misspecification, the multiplier bootstrap is known to also approximate the asymptotic normal limiting distribution with the correct sandwich covariance structure \citep{vaart1996weak}.
\end{remark}

\section{Theoretical Results}\label{sec:theory}

Our estimator and its confidence sets are built from the estimated score function $\sh(\theta, \Xn)$, and their accuracy depends on how closely $\sh(\cdot, \Xn^*)$ approximates the true score $s^*(\cdot, \Xn^*)$. We begin by introducing conditions that characterize this approximation. Let $\|A\|_{\rm F} =\sqrt{\sum_{i,j}A_{ij}^2}$ denote the matrix Frobenius norm. 

\begin{assumption}[Uniform score-matching error]\label{ass:uniform_sm_err_single}
In the neighborhood $\mathcal B(\theta^\ast; r_0)$ for some $r_0>0$, there exist sequences $\varepsilon_{i,n}\to 0$ as $n\to\infty$ for $i=1,2,3$, such that the score-matching error, curvature-matching error, and mean-matching error are uniformly bounded as
{\footnotesize\begin{align*}
& \varepsilon_{1,n}^2:= \sup_{\theta\in B(\theta^\ast; r_0)} \E_{X\sim p_{\theta^\ast}} \big[\norm{\widehat s(\theta, X)- s^\ast(\theta, X)}^2\big]\\
& \varepsilon_{2,n}^2:=\sup_{\theta\in B(\theta^\ast; r_0)} \norm{\E_{X\sim p_{\theta}} \big[ \nabla_\theta \widehat s(\theta, X) +\widehat s(\theta, X)\widehat s(\theta, X)^T \big] }_{\rm F}^2 \\
& \varepsilon_{3,n}^2:=\sup_{\theta\in B(\theta^\ast; r_0)} \norm{\E_{X\sim p_{\theta}} \big[\widehat s(\theta, X)\big]}^2.
\end{align*}}
\vspace{-1em}
\end{assumption}

Since the structured score matching procedure targets the single-data score $s^\ast(\theta,X_i)$ rather than the full-data score $s^\ast(\theta,\Xn)$, the estimation errors do not explicitly depend on the sample size $n$. However, for the overall estimation error to vanish as $n$ grows, these errors themselves must decay to zero. Fortunately, with sufficiently many simulated datasets, the three errors can be made arbitrarily small. Specifically, the score-matching error $\varepsilon_{1,n}$ can be characterized using generalization results for score estimation \citep{oko2023diffusion}. The curvature-matching error $\varepsilon_{2,n}$ can be controlled through an appropriate choice of regularization. The mean-matching error $\varepsilon_{3,n}$ stems from Monte Carlo approximation of the expected score and is of order $O(m^{-\frac{1}{2}})$, where $m$ is the Monte Carlo sample size.

\begin{theorem}[Existence, Uniqueness and Consistency]\label{thm:consistency}
Under Assumption~\ref{ass:uniform_sm_err_single} and other mild regularity assumptions in the Appendix, it holds with probability converging to $1$ as $n\to\infty$ that $\sh(\cdot, \Xn^\ast)$ has a unique root $\zero$ within $\mathcal B(\theta^\ast; r_0)$, and $\zero \to\theta^\ast$ in probability as $n\to\infty$.
\end{theorem}
It is worth highlighting that our proof of the existence of a root of $\sh(\cdot, \Xn^\ast)=0$ does not follow the standard approach, because $\sh(\cdot, \Xn^\ast)$ is not necessarily the gradient of any function. To this end, we construct an auxiliary objective function $g(\theta) = \|\sh(\theta,\Xn^\ast)\|^2$ and show that $g(\theta)$ admits at least one local minimum $\zero$ in $B(\theta^\ast; r_0)$. The first-order optimality condition, together with the nonsingularity of $\nabla_\theta \sh(\theta,\Xn)$ over $B(\theta^\ast; r_0)$, then implies $\sh(\zero,\Xn^\ast)=\bm{0}$.

\begin{theorem}[$\zero$ is close to $\MLE$]\label{thm:asymp_plugin}
Under the same assumptions as \Cref{thm:consistency}, it holds with probability converging to $1$ as $n\to\infty$ that the difference between our estimator $\zero$ in \eqref{eq:shat_root} and  MLE $\MLE$ defined in \eqref{eq:MLE_root} can be bounded as
\begin{align*}
\big\|\,\zero - \MLE\big\| \leq C\,(n^{-1/2}\varepsilon_{1,n} +n^{-1/2}\varepsilon_{2,n} +  \varepsilon_{3,n} \big).
\end{align*}
\end{theorem}

\begin{theorem}[Asymptotic normality of $\zero$]\label{cor:asymp_plugin} Under the assumptions of \Cref{thm:consistency}, 
if the mean-matching error satisfies $\varepsilon_{3,n} = o(n^{-1/2})$ as $n\to\infty$,  then our estimator $\zero$ satisfies
\[
\sqrt{n}(\zero - \theta^\ast) \overset{d}{\to} \mathcal{N}\big(0, \big[I(\theta^\ast)\big]^{-1}\big), \text{ as }  n\to \infty.
\]
\end{theorem}

To characterize bootstrap consistency, we adopt the standard notion of convergence in distribution in probability~\citep{vaart1996weak}, which we denote by $\overset{d^\ast}{\to}$. Concretely, for random probability measures $Q_n^\ast$ (e.g., bootstrap laws) and a fixed law $Q$, we write $Q_n^\ast \overset{d^\ast}{\to} Q$ if $d_{\mathrm{BL}}(Q_n^\ast, Q) \overset{p}{\to} 0$, $n\to\infty$, where $d_{\mathrm{BL}}$ is the bounded Lipschitz distance.
\begin{theorem}[Bootstrap consistency]\label{thm:asymp_boot}
% Given the data $\Xn^\ast$, and 
Under the same assumptions as \Cref{cor:asymp_plugin}, the multiplier bootstrap estimator $\boot$ defined in \eqref{eq:boot_est} satisfies 
\begin{align*}
\sqrt{n} (\boot - \zero) \overset{d^\ast}{\to} \mathcal{N}\big(0, \big[I(\theta^\ast)\big]^{-1}\big), \ \ \mbox{as }n\to\infty.
\end{align*}  
\end{theorem}
Since the limiting distribution of $\sqrt{n}(\boot - \zero)$ coincides with that of $\sqrt{n}(\zero - \theta^\ast)$, this justifies the use of bootstrap empirical quantiles to construct asymptotically valid confidence sets.

\begin{theorem}[Algorithmic convergence]\label{thm:GA_conv}
Under the same conditions as \Cref{cor:asymp_plugin}, if the initial value $\theta^{(0)}$ is drawn from the neighborhood $ \mathcal B(\theta^\ast; r_0)$ in \Cref{ass:uniform_sm_err_single},
and additionally assume $h_{\min}\,\bm{I}_{d_\theta} \preceq H_t \preceq h_{\max}\,\bm{I}_{d_\theta}$ for positive constants $(m,L,h_{\min},h_{\max})>0$. Then the iterates of Algorithm in \eqref{eq:GA} satisfy
\begin{align*}
    \big\|\theta^{(t)} - \zero \big\| \leq c\, \rho^t\, \big\|\theta^{(0)} - \zero \big\|, \quad t\geq 0,
\end{align*}
for some constant $c>0$ and $\rho = \max\{|1- c_1\alpha  h_{\min}|,\, |1-c_2\alpha h_{\max}|\}$,
where the constants $(c_1, c_2)$ only depend on the condition number of $I(\theta^\ast)$.
\end{theorem}

According to this theorem, Algorithm~\eqref{eq:GA} mimics the gradient ascent algorithm for maximizing the log-likelihood function $\log p_\theta^{(n)}(\Xn^\ast)$, and achieves exponential convergence for any stepsize $\alpha \in \big(0,(c_2 h_{\max})^{-1}\big)$.

\section{Empirical Results}\label{sec:simulation}
In this section, we evaluate the performance of the proposed method on both simulated and real datasets. Additional examples are deferred to \Cref{sec:additional_simu}.
We compare against two representative LFI methods, namely Synthetic Likelihood (SL) \citep{wood2010statistical} and the Neural Likelihood Estimator (NLE) \citep{papamakarios2019sequential}. 
Iterative Local Score Matching (ILSM) \citep{khoo2025direct} was also evaluated. Its performance exhibited greater variability (cf.~\Cref{sec:toy}), and it is therefore not included in the current table.
For our method, we compute a point estimator and four types of confidence sets: Option~1 based on plug-in estimates in \eqref{eq:I_ss} (SS) and \eqref{eq:I_hat} (Curv), Option~2 (Sand), and Option~3 (Boot). In the simulation studies, we compare all methods using the average absolute estimation error, and for our method we additionally report frequentist coverage rates and confidence interval widths. All results are averaged over 100 experiments. Full implementation details are provided in \Cref{sec:implementation_details}.

In these experiments, our method achieves more accurate point estimation than competing approaches at lower simulation cost, while also providing valid frequentist confidence intervals. Moreover, the curvature matching property of our method enables the use of second-order optimization, such as quasi-Newton methods, yielding faster convergence compared to first-order approaches. 

\subsection{M/G/1-queueing Model}\label{sec:mg1}
Our first example is the M/G/1 queuing model, a standard benchmark in LFI. The model is governed by three parameters $\theta = (\theta_1, \theta_2, \theta_3)$ and generates customer inter-departure times in a single-server setting.
 Following the specification in \citet{jiang2018approximate}, the dataset consists of $500$ independent sequences, each of length $5$, where an observation takes the form $x_i = (x_{i1},\ldots,x_{i5})^T$. Service times are sampled as $u_{ik} \sim U[\theta_1,\theta_2]$ and arrival times as $w_{ik}\sim \text{Exp}(\theta_3)$. The inter-departure times are then determined recursively by $x_{ik} =u_{ik} + \max(0,\sum^k_{j=1} w_{ij}-\sum^{k-1}_{j=1} x_{ij})$. 
 We generate the observed dataset $\Xn^\ast$ under the ground truth $\theta^\ast = (1, 5, 0.2)$. For our method and NLE, we use a uniform distribution on $[0,10] \times [0.01,10] \times [0.01,0.5]$ as the sampling distribution for $(\theta_1, \theta_2 - \theta_1, \theta_3)$.

\begin{table*}[!ht]
\centering
\caption{Average estimation errors, and coverage and width of 95\% CIs under the M/G/1-queuing model, with standard deviations. Highlighted coverages are near nominal and with no substantial increase in width.}
\label{tab:mgq_merged}
\resizebox{\textwidth}{!}{
\begin{tabular}{l|ccc|cc|cc|cc}
\toprule
& \multicolumn{3}{c|}{Estimation error}
& \multicolumn{2}{c|}{$\theta_1^\ast = 1$}
& \multicolumn{2}{c|}{$\theta_2^\ast = 5$}
& \multicolumn{2}{c}{$\theta_3^\ast = 0.2$} \\
& $\lvert \widehat{\theta}_1 - \theta_1^\ast \rvert$
& $\lvert \widehat{\theta}_2 - \theta_2^\ast \rvert$
& $\lvert \widehat{\theta}_3 - \theta_3^\ast \rvert$
& Cover & Width & Cover & Width & Cover & Width \\
& & & (scale $\times 10^{-2}$) & & & & & & \\
\midrule
SS
& \multirow{4}{*}{$0.037$ ($0.030$)}
& \multirow{4}{*}{$\mathbf{0.100}$ ($0.076$)}
& \multirow{4}{*}{$\mathbf{0.390}$ ($0.303$)}
& 0.82 & $0.129$ ($0.020$) & 0.90 & $0.410$ ($0.048$) & 0.91 & $0.015$ ($0.001$) \\
Curv
& & & 
& 0.86 & $0.144$ ($0.016$) & 0.91 & $0.438$ ($0.040$) & 0.92 & $0.016$ ($0.001$) \\
Sand
& & &
& \textbf{0.92} & $0.161$ ($0.017$) & \textbf{0.94} & $0.471$ ($0.046$) & \textbf{0.93} & $0.018$ ($0.001$) \\
Boot
& & &
& \textbf{0.92} & $0.161$ ($0.018$) & \textbf{0.94} & $0.470$ ($0.046$) & \textbf{0.93} & $0.018$ ($0.001$) \\
\midrule
NLE
& $\mathbf{0.033}$ ($0.023$) & $0.133$ ($0.091$) & $0.404$ ($0.313$)
& 0.80 & $0.091$ ($0.016$) & 0.77 & $0.308$ ($0.143$) & 0.90 & $0.017$ ($0.001$) \\
SL
& $0.504$ ($0.328$) & $0.626$ ($0.415$) & $0.442$ ($0.364$)
& 0.95 & $2.098$ ($0.471$) & 0.95 & $2.794$ ($0.559$) & \textbf{0.96} & $0.022$ ($0.003$) \\
\bottomrule
\end{tabular}
}
\end{table*}

We present the absolute errors, the coverage rate and width of the $95\%$ confidence interval of all the methods in \Cref{tab:mgq_merged}. The result shows that our method and NLE perform the best in terms of point estimation for this example. 
The result of the confidence interval indicates that the bootstrap method and sandwich form covariance perform the best.
However, the CI using the plug-in covariance \eqref{eq:I_hat} or \eqref{eq:I_ss} yields inferior performance in this example.

\begin{figure}[!ht]
    \centering
    \includegraphics[width=0.49\textwidth]{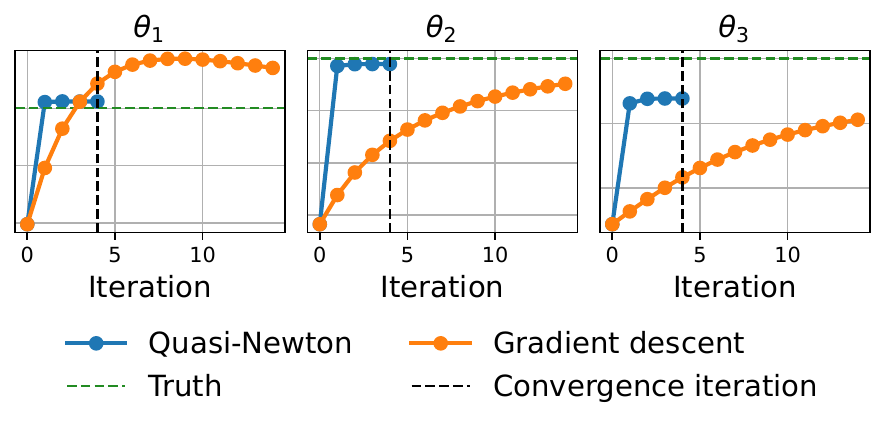}
    \caption{Comparison of the Quasi-Newton algorithm and vanilla gradient descent}
    \label{fig:quasi-newton_cutoff}
\end{figure}

We also compare the convergence of the Quasi-Newton method with vanilla gradient descent in this example. Starting from the estimate obtained in the first round (see Remark~\ref{rem:two}), we use the estimated Hessian at the initial iterate as the preconditioning matrix. Each algorithm is terminated once the change of the iterates falls below $10^{-6}$ in Euclidean norm. Over $100$ replicates, the average numbers of iterations to convergence are $3.97$ for Quasi-Newton and $81.13$ for gradient descent, with standard deviation $0.89$ and $19.60$ respectively. We present a typical trace plot in \Cref{fig:quasi-newton_cutoff}, showing that even a single Quasi-Newton step yields a solution close to convergence.

\subsection{Stock Volatility Estimation}
We further apply our method to a stock volatility estimation problem \citep{wang2022approximate, magdon2003maximumlikelihood}.
We follow their setup and model the log price of two stocks $X(t) = (X_1(t), X_2(t))$ using a multivariate Brownian motion
\begin{align}
dX(t) = \begin{pmatrix}
\mu_1\\
\mu_2
\end{pmatrix} dt+ \begin{pmatrix}
\sigma_1^2 & \rho \sigma_1\sigma_2 \\
\rho \sigma_1\sigma_2 & \sigma_2^2
\end{pmatrix} dW(t),
\end{align}
and we re-parameterize as 
\[(\mu_1, \mu_2, \sigma_1, \sigma_2, \rho) = (\theta_1, \theta_2, \exp(\theta_3), \exp(\theta_4), \tanh(\theta_5)).\]  
The observed data only consists of the high, low and closing log prices of the two stocks.
Following \citet{wang2022approximate}, we use a true parameter $(\theta_1^\ast, \theta_2^\ast, \theta_3^\ast, \theta_4^\ast, \theta_5^\ast) = (0, 0, 0, 0, \text{arctanh}(0.5))$ to generate i.i.d. observed data of sample size $1000$ and run all the methods, where a uniform proposal distribution on $[-3, 3]^2\times[-3, 2]^2 \times [-3, 3]$ is adopted for our method and NLE. 

We repeat the experiment $100$ times, and present the results in \Cref{tab:volatility_err,tab:volatility_cover}, showing that our method has the smallest estimation error for all the parameters, and the CIs constructed by plug-in estimates or by bootstrap all have near-nominal coverage rate and smaller width than those of NLE or SL. 

\begin{table*}[!ht]
\centering
\caption{Average estimation errors with standard deviation under the Stock Volatility Model}
\label{tab:volatility_err}
\resizebox{0.71\textwidth}{!}{
\begin{tabular}{l|c|c|c|c|c}
\toprule
 & $\lvert \hat\theta_1 - \theta_1^\ast \rvert$ & $\lvert \hat\theta_2 - \theta_2^\ast \rvert$ & $\lvert \hat\theta_3 - \theta_3^\ast \rvert$ & $\lvert \hat\theta_4 - \theta_4^\ast \rvert$ & $\lvert \hat\theta_5 - \theta_5^\ast \rvert$ \\
\midrule
Ours & $\mathbf{0.026}$ ($0.019$) & $\mathbf{0.023}$ ($0.021$) & $\mathbf{0.006}$ ($0.005$) & $\mathbf{0.007}$ ($0.005$) & $\mathbf{0.018}$ ($0.012$) \\
NLE & $0.038$ ($0.028$) & $0.068$ ($0.040$) & $0.007$ ($0.006$) & $0.023$ ($0.011$) & $0.031$ ($0.023$) \\
SL & $0.254$ ($0.213$) & $0.252$ ($0.199$) & $0.205$ ($0.137$) & $0.203$ ($0.139$) & $0.570$ ($0.338$) \\
\bottomrule
\end{tabular}
}
\end{table*}

\begin{table*}[!ht]
\caption{Averaged coverage and width with standard deviation of 95\% CIs under the Stock Volatility Model. Highlighted values indicate coverage near nominal with no substantial increase in width.}
\label{tab:volatility_cover}
\resizebox{\textwidth}{!}{
\begin{tabular}{l|cc|cc|cc|cc|cc}
\toprule
 & \multicolumn{2}{c|}{$\theta_1^\ast = 0$} & \multicolumn{2}{c|}{$\theta_2^\ast = 0$} & \multicolumn{2}{c|}{$\theta_3^\ast = 0$} & \multicolumn{2}{c|}{$\theta_4^\ast = 0$} & \multicolumn{2}{c}{$\theta_5^\ast = 0.55$} \\
 & Cover & Width & Cover & Width & Cover & Width & Cover & Width & Cover & Width \\
\midrule
SS & \textbf{0.94} & $0.123$ ($0.002$) & 0.90 & $0.122$ ($0.003$) & \textbf{0.95} & $0.031$ ($0.001$) & \textbf{0.92} & $0.032$ ($0.001$) & \textbf{0.95} & $0.083$ ($0.003$) \\
Curv & \textbf{0.94} & $0.123$ ($0.001$) & \textbf{0.92} & $0.123$ ($0.001$) & \textbf{0.94} & $0.031$ ($0.000$) & \textbf{0.92} & $0.031$ ($0.000$) & \textbf{0.95} & $0.083$ ($0.001$) \\
Sand & \textbf{0.95} & $0.124$ ($0.003$) & \textbf{0.92} & $0.124$ ($0.003$) & \textbf{0.93} & $0.031$ ($0.001$) & \textbf{0.93} & $0.031$ ($0.001$) & \textbf{0.95} & $0.083$ ($0.002$) \\
Boot & \textbf{0.95} & $0.124$ ($0.004$) & 0.91 & $0.124$ ($0.004$) & \textbf{0.96} & $0.031$ ($0.001$) & \textbf{0.92} & $0.031$ ($0.001$) & \textbf{0.95} & $0.082$ ($0.003$) \\
\midrule
NLE & 0.91 & $0.157$ ($0.022$) & 0.87 & $0.251$ ($0.063$) & 0.95 & $0.036$ ($0.003$) & 0.34 & $0.038$ ($0.003$) & 0.84 & $0.103$ ($0.005$) \\
SL & 0.99 & $3.546$ ($1.273$) & 0.99 & $3.432$ ($1.316$) & 0.99 & $1.097$ ($0.688$) &0.97 &$1.121$ ($0.693$) &0.97 &$3.394$ ($0.792$)\\
\bottomrule
\end{tabular}
}
\end{table*}

\subsection{The g-and-k Model}
The g-and-k model is widely used in robust statistical modeling because of its ability to flexibly capture non-normal features such as skewness and heavy tails while remaining relatively simple to simulate from. It is commonly applied in fields such as finance \citep{drovandi2011likelihood}, insurance \citep{peters2016estimating}, and risk modeling, where data often exhibit asymmetry and extreme values that violate Gaussian assumptions. It is defined implicitly through its quantile function $P^{-1}_\theta(u)= Q(z(u);\theta)$ where $z(\cdot)$ is the quantile of $\mN(0,1)$ and $\theta=(A,B,g,k)$. The $Q$ function is defined as
\[
Q(z; A,B, g, k)= A+B\Big[1+c\cdot \text{tanh}(gz/2)\Big] z(1+z^2)^k.
\] 
where parameters $A,B,g,k$ control location, scale, skewness, and tail heaviness, respectively.  Despite its flexibility and interpretability, the likelihood function of the g-and-k model is intractable. The simulation study on the g-and-k model is provided in \Cref{sec:g-k-simulation}.

Here we illustrate the performance of our method to the \texttt{Garch} exchange rate dataset from the Ecdat package \citep{crossiant2016ecdat}. It consists of 1967 daily US dollar exchange rates against other currencies from 1980 to 1987. We focus on the exchange rate with Canadian Dollars. 
We use a g-and-k distribution to model the log return.
% , which is defined as $\log(X_{t+1}/X_t)$. 
We fit the dataset using our two-round score matching procedure, with more details provided in \Cref{sec:implementation_details}. 
We evaluate the model fit using a Q–Q plot that compares the observed data with samples generated from the fitted g-and-k model and from a fitted normal baseline in \Cref{fig:QQ_g-and-k}. 

\begin{figure}[!ht]
    \centering
    \includegraphics[width=0.45\textwidth]{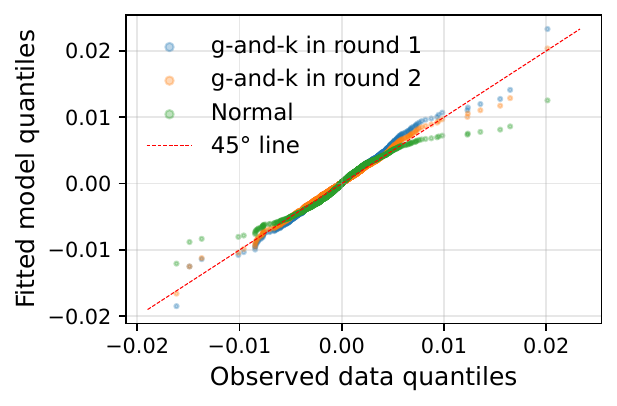}
    \caption{Q--Q plot comparing the fitted g-and-k/normal distribution with the observed data}
    \label{fig:QQ_g-and-k}
\end{figure}

As shown in \Cref{fig:QQ_g-and-k}, the g-and-k distribution provides a much better fit to this heavy-tailed log-return data than the normal distribution. Furthermore, the second-round estimates show a noticeable improvement over the first round, confirming the effectiveness of the two-round procedure.

\section{Discussion} \label{sec:discussion}
This work proposes a new framework for likelihood-free inference by using the basic fact that the MLE is the root of the likelihood score. By approximating the score rather than the likelihood itself, we build estimators and confidence sets that remain closely connected to likelihood-based inference. In particular, by adding statistical structure into the score network, we improve score estimation and enhance uncertainty quantification.

For future work, one promising direction is to add additional structural constraints such as sparsity or symmetry to the score Jacobian, which could further improve estimation. We are also interested in extending the method beyond the i.i.d.~setting, for example to dependent time-series or spatial data, and in adapting the approach to regression problems.

\subsubsection*{Acknowledgements}
We thank the the anonymous reviewers for their constructive
comments and suggestions. YW gratefully acknowledges support from National Science Foundation (DMS-2515542).

\FloatBarrier 
\bibliographystyle{apalike}
\bibliography{MLE}

\clearpage
\appendix
\thispagestyle{empty}

\section{Proofs}

Before presenting the proofs of the theorems in the main paper, we introduce some common notation and regularity assumptions required for the results.

\paragraph{Notations} We use $\sigma_{\min}(A)$ to denote the smallest singular value of a matrix $A$. 
In the proofs, we frequently apply Taylor expansions to each coordinate of 
$\sh(\cdot, \Xn^\ast)$ and $\nabla \sh(\cdot, \Xn^\ast)$, 
that is, to $\sh_j(\cdot, \Xn^\ast): \R^{d_\theta} \to \R$ and 
$\nabla_\theta \sh_{ij}(\cdot, \Xn^\ast): \R^{d_\theta} \to \R$, 
both viewed as scalar-valued functions of $\theta$. 
Note that when performing element-wise Taylor expansions, 
the remainder term may be evaluated at different intermediate parameter values 
for different components of the function.
Let $\theta_A, \theta_B \in \R^{d_\theta}$ denote two generic parameter vectors. We will consider three types of Taylor expansions:
\begin{itemize}
\item \textbf{Scenario 1:}
\[
\sh(\theta_B, \Xn^\ast) 
= \sh(\theta_A, \Xn^\ast) 
+ H_n(\sh, \theta_A, \theta_B)(\theta_B - \theta_A),
\]
where 
\begin{align}\label{eq:Taylor_1}
H_n(\sh, \theta_A, \theta_B)
:= 
\begin{pmatrix}
\nabla (\sh - s^\ast)_1 (\theta_1', \Xn^\ast)\\
\vdots \\
\nabla (\sh - s^\ast)_{d_\theta}(\theta_{d_\theta}', \Xn^\ast)
\end{pmatrix}
\in \R^{d_\theta\times d_\theta},
\end{align}
and each $\theta_j' = c_j \theta_A + (1-c_j)\theta_B$ for some $c_j \in [0,1]$.

\item \textbf{Scenario 2:}
\begin{align}\label{eq:Taylor_2}
\sh(\theta_B, \Xn^\ast) 
= \sh(\theta_A, \Xn^\ast) 
+ \nabla_\theta \sh(\theta_A, \Xn^\ast)(\theta_B - \theta_A) 
+ \frac{1}{2}T_n(\sh, \theta_A, \theta_B)[\theta_B - \theta_A, \theta_B - \theta_A],
\end{align}
where $T_n(\sh, \theta_A, \theta_B) \in \R^{d_\theta\times d_\theta\times d_\theta}$, 
and its $i$-th slice is the Hessian matrix $\nabla_\theta^2 \sh_i(\theta_i', \Xn^\ast)$ 
evaluated at $\theta_i' = c_i \theta_A + (1-c_i)\theta_B$ for some $c_i \in [0,1]$.  
Here, for $z \in \R^d$ and a tensor $A = [A_1, \dots, A_p] \in \R^{p\times d\times d}$ 
with $A_1,\dots,A_p \in \R^{d\times d}$, we denote
\[
A[z, z] := (z^\top A_1 z, \dots, z^\top A_p z)^\top, 
\quad 
A[z] := (A_1^\top z, \dots, A_p^\top z)^\top.
\]

\item \textbf{Scenario 3:}
\begin{align}\label{eq:Taylor_3}
\nabla_\theta \sh(\theta_B, \Xn^\ast) 
= \nabla_\theta \sh(\theta_A, \Xn^\ast) 
+ T_n(\sh, \theta_A, \theta_B)[\theta_B - \theta_A],
\end{align}
where $T_n(\sh, \theta_A, \theta_B) \in \R^{d_\theta \times d_\theta \times d_\theta}$, 
and its $(i,j)$-th element is 
\[
\nabla_\theta \!\left\{ 
\frac{\partial \sh_j(\theta, \Xn^\ast)}{\partial \theta_{[i]}} 
\right\} 
\in \R^{d_\theta},
\]
evaluated at $\theta_{i,j}' = c_{i,j}\theta_A + (1-c_{i,j})\theta_B$ for some $c_{i,j}\in[0,1]$.  
Here, $\theta_{[i]}$ denotes the $i$-th coordinate of $\theta$.  
The operation $T_n(\sh, \theta_A, \theta_B)[\theta_B - \theta_A]$ 
computes the inner product between $(\theta_B - \theta_A)$ and each $(i,j)$-th element of 
$T_n(\sh, \theta_A, \theta_B)$.
\end{itemize}

Now we impose additional regularity assumptions on the model and the likelihood score functions.
\begin{assumption}[Nonsingular Fisher information]\label{ass:curv_cond}
$\sigma_{\min} \big(I(\theta^\ast)\big) \geq  m$ for some $m > 0$.
\end{assumption}
Since the true Fisher information matrix $I(\theta^\ast)$ is always positive semi-definite, this condition ensures that $I(\theta^*) \succeq mI$, i.e., $\big(I(\theta^*) - mI\big)$ is positive semi-definite. Such a condition is commonly imposed to guarantee that the true model parameter is identifiable.

\begin{assumption}[Regularity conditions on Score]\label{ass:reg_cond}
For some $\delta > 0$  and any $x$ 
\begin{align*}
\sup_{\theta \in \mathcal B(\theta^\ast; \delta)} \lVert \sh(\theta, x) \rVert^2 \leq G(x), \sup_{\theta \in \mathcal B(\theta^\ast; \delta)} \lVert \nabla_\theta \sh(\theta, x) \rVert_F^2 \leq G(x),\sup_{\theta \in \mathcal B(\theta^\ast; \delta)} \lVert \nabla_\theta^2 \sh_j(\theta, x) \rVert_F \leq G(x) 
\end{align*}
where $\sh_j(\theta, X)$ denotes the $j$-th coordinate of $\sh(\theta, X)$ and the function $G(\cdot)$ satisfies $\E_{X\sim P_{\theta^*}} [G(X)]<C_3^2$ for some constant $C_3>0$. Additionally, we assume the estimated score  is Lipschitz in $\theta$, i.e.,
\begin{align*}
&\lVert \sh(\theta, X) - \sh(\theta^\ast, X) \rVert \leq L(X) \lVert \theta - \theta^\ast \rVert ,\quad  \E_{X\sim P_{\theta^\ast}} [L(X)^2] < \infty, \forall \theta\in \mathcal B(\theta^\ast; \delta).
\end{align*}
Moreover, 
\[
\E_{X\sim P_{\theta^\ast}} \lVert \sh(\theta^\ast, X) \rVert^{2+\gamma} <\infty
\]
for some $\gamma > 0$.
\end{assumption}
The first part of this assumption consists of standard regularity conditions on the likelihood score for parametric models, ensuring that the MLE (or in our context, $\zero$) is consistent and asymptotically normal. The last bounded $(2+\gamma)$ moment condition is used to verify Lindeberg’s condition, which guarantees the asymptotic normality of the bootstrap estimator.

The following lemma shows that the difference between the estimated Fisher information matrix $\wh I(\theta^*)$ and the true Fisher information matrix $I(\theta^*)$ can be bounded in terms of the score-matching and curvature-matching errors.
\begin{lemma}[Lemma 8 of \citet{jiang2025simulation}]\label{lem:curv_err_single}
Under Assumptions \ref{ass:uniform_sm_err_single} and \ref{ass:reg_cond}, we have
\[
\big\lVert I(\theta^\ast) - \widehat{I}(\theta^\ast) \big\rVert_2 \le 2C_3\varepsilon_{1,n} + \varepsilon_{2,n}
\]
where $C_3>0$ is as defined in \Cref{ass:reg_cond}.
\end{lemma}

This bound ensures that $\Ih(\theta^\ast)$ preserves the invertibility property of $I(\theta^\ast)$, even though it is not necessarily symmetric, as shown below.
\begin{lemma}\label{lem:curv_Ihat}
Under Assumptions \ref{ass:uniform_sm_err_single}, \ref{ass:curv_cond} and \ref{ass:reg_cond}, we have that for any $y \in \R^{d_\theta}$,
\[
\lVert \Ih(\theta^\ast)\, y\rVert \,\geq\, (m - 2C_3\varepsilon_{1,n} - \varepsilon_{2,n}) \,\lVert y \rVert.
\]
\end{lemma}
\begin{proof}
We have the following inequalities,
\begin{align*}
\lVert \Ih(\theta^\ast) y\rVert & \geq \lVert I(\theta^\ast) y\rVert -  \lVert [\Ih(\theta^\ast)-I(\theta^\ast)] y\rVert \\
&\geq (m - 2C_3\varepsilon_{1,n} - \varepsilon_{2,n}) \lVert y \rVert
\end{align*}
where the first step is by triangle inequality and the second step is by \Cref{ass:curv_cond} and \Cref{lem:curv_err_single}. 
    
\end{proof}

\subsection{Proof of \Cref{thm:consistency}}\label{sec:proof_consistency}
We consider any $r_0 < \frac{m}{2C_3}$ (specified in Assumption \ref{ass:uniform_sm_err_single}) through out the proof.
\paragraph{Part 1: Existence of the root.}
We first show the existence of the root within $\mathcal B(\theta^\ast; r_0)$. Since $\mathcal B(\theta^\ast; r_0)$ is compact and $\sh(\cdot, \Xn^\ast)$ is continuous, the objective function
\[
F(\theta) := \lVert \sh(\cdot, \Xn^\ast) \rVert^2
\]
has a minimizer $\bar{\theta}$ within the compact set $\mathcal B(\theta^\ast; r_0)$. If $\bar{\theta}$ does not belong to $\partial \mathcal B(\theta^\ast; r_0)$, the boundary of $\mathcal B(\theta^\ast; r_0)$, then the first order optimal condition gives
\[
\nabla_\theta F(\bar \theta) = 2\nabla_\theta\sh(\bar{\theta}, \Xn^\ast) \sh(\bar{\theta}, \Xn^\ast) = \bm 0,
\]
which further yields $\sh(\bar{\theta}, \Xn^\ast) = \bm 0$ if $\nabla_\theta\sh(\bar{\theta}, \Xn^\ast)$ is invertible.
Therefore, it suffices to show $\bar{\theta} \notin \partial \mathcal B(\theta^\ast; r_0)$ and $\nabla_\theta\sh(\theta, \Xn^\ast)$ is invertible for any $\theta\in \mathcal B(\theta^\ast; r_0)$. For the first part, we can verify it by showing $F(\MLE) < F(\theta)$ for any $\theta \in \partial \mathcal B(\theta^\ast; r_0)$ and the fact that $\MLE$ is in the interior of $\mathcal B(\theta^\ast; r_0)$.  To prove this, notice that by applying the Taylor expansion for each element of $\sh_j$, we can obtain
\begin{align}\label{eq:sconorm_MLE}
&\Big\lVert \frac{1}{n}\sh\Big(\MLE, \Xn^\ast\Big) \Big\rVert \notag \\
=& \,\Big\lVert \frac{1}{n} \sh\Big(\MLE, \Xn^\ast\Big) - \frac{1}{n} s^\ast\Big(\MLE, \Xn^\ast\Big) \Big\rVert \notag \\
=& \,\Big\lVert \frac{1}{n} \sh(\theta^\ast, \Xn^\ast) - \frac{1}{n} s^\ast(\theta^\ast, \Xn^\ast) + \frac{1}{n}H_n(\sh - s^\ast, \theta^\ast, \MLE)(\MLE - \theta^\ast) \Big\rVert \notag \\
\leq& \,\Big\lVert \frac{1}{n} \sh(\theta^\ast, \Xn^\ast) - \frac{1}{n} s^\ast(\theta^\ast, \Xn^\ast) \Big\rVert + \lVert \frac{1}{n} H_n(\sh - s^\ast, \theta^\ast, \MLE) \rVert \cdot \Big\lVert \MLE - \theta^\ast \Big\rVert \notag \\
\leq& \,\big\lVert \E_{X\sim P_{\theta^\ast}} \big[\sh(\theta^\ast, X) \big] \big\rVert + o_{P_{\theta^\ast}}(1) + C_3 \sqrt{\frac{\log n}{n}} \notag \\
\leq& \,\varepsilon_{3, n} + o_{P_{\theta^\ast}}(1),
\end{align}
where $H_n(\sh - s^\ast, \theta^\ast, \MLE)$ is defined analogously to \eqref{eq:Taylor_1}. The second-to-last step follows from the law of large numbers, \Cref{ass:reg_cond}, and the classical result on the maximum likelihood estimator $\MLE$ that $\MLE \in \mathcal B(\theta^\ast; C\sqrt{\frac{\log n}{n}})$ with probability tending to $1$, for some constant $C>0$.

Similarly, for any $\theta \in \partial \mathcal B(\theta^\ast; r_0)$, we have
\begin{align}\label{eq:sconorm_bdr}
&\lVert \frac{1}{n}\sh(\theta, \Xn^\ast) - \frac{1}{n}s^\ast(\theta^\ast, \Xn^\ast) \rVert \notag \\
=& \,\lVert \frac{1}{n}\sh(\theta^\ast, \Xn^\ast) - \frac{1}{n}s^\ast(\theta^\ast, \Xn^\ast) + \frac{1}{n}\nabla_\theta \sh(\theta^\ast, \Xn^\ast) (\theta - \theta^\ast) + \frac{1}{2n} T_n(\sh, \theta^\ast, \theta)[\theta - \theta^\ast, \theta - \theta^\ast] \rVert  \notag \\
\geq& \,\lVert \frac{1}{n}\nabla_\theta \sh(\theta^\ast, \Xn^\ast) (\theta - \theta^\ast) \rVert  - \lVert \frac{1}{n}\sh(\theta^\ast, \Xn^\ast) - \frac{1}{n}s^\ast(\theta^\ast, \Xn^\ast) +  \frac{1}{2n} T_n(\sh, \theta^\ast, \theta)[\theta - \theta^\ast, \theta - \theta^\ast] \rVert \notag \\
\geq&\, (m - 2C_3\varepsilon_{1,n}-\varepsilon_{2,n})r_0 - \varepsilon_{3,n} - C_3r_0^2 + o_{P_{\theta^\ast}}(1) \notag \\
\Rightarrow\quad& \lVert \frac{1}{n}\sh(\theta, \Xn^\ast)\rVert \geq (m - 2C_3\varepsilon_{1,n}-\varepsilon_{2,n}) r_0 - \varepsilon_{3,n} - C_3r_0^2 - \lVert \frac{1}{n}s^\ast(\theta^\ast, \Xn^\ast) \rVert + o_{P_{\theta^\ast}}(1) \notag \\
\quad& \,\lVert \frac{1}{n}\sh(\theta, \Xn^\ast)\rVert \geq (m - 2C_3\varepsilon_{1,n}-\varepsilon_{2,n}) r_0 - \varepsilon_{3,n} - C_3r_0^2 + o_{P_{\theta^\ast}}(1)
\end{align}
where the first step follows from the Taylor expansion, and $T_n(\sh, \theta^\ast, \theta)$ is defined in \eqref{eq:Taylor_2}; the second and fourth steps follow from the triangle inequality; the third step follows from \Cref{lem:curv_Ihat}, the law of large numbers, \Cref{ass:reg_cond}, and the triangle inequality; and the last step follows from the law of large numbers. If $2C_3r_0\varepsilon_{1,n} + r_0\varepsilon_{2,n} + 2\varepsilon_{3,n} < r_0(m - C_3r_0)$, then combining \eqref{eq:sconorm_MLE} and \eqref{eq:sconorm_bdr} yields $\big\lVert \frac{1}{n}\sh(\MLE, \Xn^\ast) \big\rVert < \big\lVert \frac{1}{n}\sh(\theta, \Xn^\ast) \big\rVert + o_{P_{\theta^\ast}}(1)$ for any $\theta \in \partial \mathcal B(\theta^\ast; r_0)$, and thus the proof of the first part is completed.

To prove the second part, i.e.~$\nabla_\theta\sh(\bar{\theta}, \Xn^\ast)$ is invertible, we use \Cref{lem:curv_err_single} and apply the Taylor expansion to $\nabla_\theta\sh(\theta, \Xn^\ast)$ for $\theta$ around $\theta^\ast$. Specifically, for any $\theta \in  \mathcal B(\theta^\ast; r_0)$, we have
\begin{align}\label{eq:curv_sh}
&\lVert \frac{1}{n} \nabla_\theta \sh(\theta, \Xn^\ast) - I(\theta^\ast) \rVert \notag \\
=& \,\lVert \frac{1}{n} \nabla_\theta \sh(\theta^\ast, \Xn^\ast) - I(\theta^\ast) + \frac{1}{n}T_n(\sh, \theta^\ast, \theta) [\theta - \theta^\ast] \rVert \notag \\
\leq& \,2C_3\varepsilon_{1,n} + \varepsilon_{2,n} + C_3r_0 + o_{P_{\theta^\ast}}(1),
\end{align}
where the notation of $T_n(\sh, \theta^\ast, \theta)$ follows \eqref{eq:Taylor_3}. Consequently, $\sigma_{\min}\!\left(\frac{1}{n}\nabla_\theta \sh(\theta, \Xn^\ast)\right) \geq m - (2C_3\varepsilon_{1,n} + \varepsilon_{2,n} + C_3r_0 + o_{P_{\theta^\ast}}(1)) > 0$ with probability tending to $1$. Hence, $\nabla_\theta \sh(\bar{\theta}, \Xn^\ast)$ is invertible with probability tending to $1$, and the proof of the existence is completed.

\paragraph{Part 2: Uniqueness of the root.} Next, the proof of uniqueness is based on the curvature property of $\sh$. If $\zero$ and $\widetilde{\theta}_n$ are both zeros of $\sh(\cdot, \Xn^\ast)$ in $\partial \mathcal B(\theta^\ast; r_0)$, then by the mean value theorem and the fundamental law of calculus, we have
\[
\bm 0 = \sh(\zero, \Xn^\ast) - \sh(\widetilde{\theta}_n, \Xn^\ast)  = \int_0^1 \nabla_\theta \sh \big(\widetilde{\theta}_n + t(\zero - \widetilde{\theta}_n), \Xn^\ast\big) \big(\zero - \widetilde{\theta}_n\big)\;dt 
\]
Then, the inequality below implies $\widetilde{\theta} = \zero$:
\begin{align*}
0 &= \frac{1}{n} \big(\zero - \widetilde{\theta}_n\big)^T \big(\sh(\zero, \Xn^\ast) - \sh(\widetilde{\theta}_n, \Xn^\ast) \big) \\
&= \frac{1}{n} \int_0^1 \big(\zero - \widetilde{\theta}_n\big)^T \nabla_\theta \sh \big(\widetilde{\theta}_n + t(\zero - \widetilde{\theta}_n), \Xn^\ast\big) \big(\zero - \widetilde{\theta}_n\big)\;dt \\
&= \int_0^1 \big(\zero - \widetilde{\theta}_n\big)^T \frac{\nabla_\theta \sh \big(\widetilde{\theta}_n + t(\zero - \widetilde{\theta}_n), \Xn^\ast\big) + \Big[\nabla_\theta \sh \big(\widetilde{\theta}_n + t(\zero - \widetilde{\theta}_n), \Xn^\ast\big) \Big]^T}{2n} \big(\zero - \widetilde{\theta}_n\big)\;dt \\
&\geq (m - 2C_3 \varepsilon_{1,n} - \varepsilon_{2, n} - C_3 r_0 + o_{P_{\theta^\ast}}(1)) \big\lVert\zero - \widetilde{\theta}_n \big\rVert^2 
\end{align*}
where the last step is due to \eqref{eq:curv_sh}.

\paragraph{Part 3: Consistency.} Finally, to prove the consistency, it suffices to show that for $\theta \in \mathcal B(\theta^\ast; r_0)$ with $ \big\lVert \theta - \theta^\ast \big\rVert > \frac{2\varepsilon_{3,n}}{m} + \sqrt{\frac{\log n}{n}} $ we always have $\big\lVert \frac{1}{n} \sh(\theta, \Xn^\ast) \big\rVert > 0$. First, by the same derivation as in \eqref{eq:sconorm_bdr}, we have
\begin{align*}
&\lVert \frac{1}{n}\sh(\theta, \Xn^\ast) - \frac{1}{n}s^\ast(\theta^\ast, \Xn^\ast) \rVert \geq m \lVert \theta - \theta^\ast \rVert - \varepsilon_{3,n} - C_3\lVert \theta - \theta^\ast \rVert^2 + O_{P_{\theta^\ast}}(\frac{1}{\sqrt{n}}) \\
\Rightarrow\quad& \lVert \frac{1}{n}\sh(\theta, \Xn^\ast) \rVert \geq \lVert \frac{1}{n}\sh(\theta, \Xn^\ast) - \frac{1}{n}s^\ast(\theta^\ast, \Xn^\ast) \rVert - \lVert \frac{1}{n}s^\ast(\theta^\ast, \Xn^\ast) \rVert \\
&\qquad\qquad\quad\;\; \geq m \lVert \theta - \theta^\ast \rVert - \varepsilon_{3,n} - C_3\lVert \theta - \theta^\ast \rVert^2 + O_{P_{\theta^\ast}}(\frac{1}{\sqrt{n}}) \\
&\qquad\qquad\quad\;\; \geq (m-C_3r_0) \lVert \theta - \theta^\ast \rVert - \varepsilon_{3,n} + O_{P_{\theta^\ast}}(\frac{1}{\sqrt{n}}).
\end{align*}
Since $r_0 < \frac{m}{2C_3}$, it follows that $\big\lVert \frac{1}{n}\sh(\theta, \Xn^\ast) \big\rVert > 0$ with probability tending to $1$, thereby establishing the conclusion.

\subsection{Proof of \Cref{thm:asymp_plugin}}\label{sec:proof_asymp_plugin}

We follow the classical proof of the asymptotic normality of the MLE by applying a Taylor expansion to $\sh(\zero, \Xn^\ast)$ and $s^\ast(\MLE, \Xn^\ast)$ around $\theta^\ast$, and analyzing how the score and curvature estimation errors affect $\sqrt{n}(\zero - \MLE)$.

By Taylor's theorem, we obtain
\begin{align}
\bm 0 &= \sh(\zero, \Xn^\ast) \notag\\
&= \sh(\theta^\ast, \Xn^\ast) + \nabla_\theta \sh(\theta^\ast, \Xn^\ast)(\zero - \theta^\ast) + \frac{1}{2}T_n(\sh, \theta^\ast, \zero) [\zero - \theta^\ast, \zero - \theta^\ast] \notag\\
\Rightarrow \sqrt{n} (\zero - \theta^\ast) &= -\Big\{\frac{1}{n}\nabla_\theta \sh(\theta^\ast, \Xn^\ast) + \frac{1}{2n}T_n(\sh, \theta^\ast, \zero)[\zero - \theta^\ast] \Big\}^{-1} \frac{1}{\sqrt n}\sh(\theta^\ast, \Xn^\ast) \notag\\
&= \Big\{ \big[\Ih(\theta^\ast)\big]^{-1} + o_{P_{\theta^\ast}}(1) \Big\} \frac{1}{\sqrt n}\sh(\theta^\ast, \Xn^\ast), \label{eq:root}
\end{align}
where the last step follows from the consistency of $\zero$ and \Cref{ass:reg_cond}, and the notation $T_n(\sh, \theta^\ast, \MLE)$ is as defined in \eqref{eq:Taylor_2}.

Similarly, for the MLE $\MLE$, we have a similar expansion:
\begin{align}
\sqrt{n} (\MLE - \theta^\ast) &= -\Big\{\frac{1}{n}\nabla_\theta s^\ast(\theta^\ast, \Xn^\ast) + \frac{1}{2n}T_n(s^\ast, \theta^\ast, \MLE)[\MLE - \theta^\ast] \Big\}^{-1} \frac{1}{\sqrt n}s^\ast(\theta^\ast, \Xn^\ast) \notag \\
&= \Big\{ \big[I(\theta^\ast)\big]^{-1} + o_{P_{\theta^\ast}}(1) \Big\} \frac{1}{\sqrt n}s^\ast(\theta^\ast, \Xn^\ast). \label{eq:mle}
\end{align}

Combining \eqref{eq:root} and \eqref{eq:mle}, we have
{\small
\begin{align*}
    &\sqrt{n}(\zero - \MLE) \\
    =&\, \Big\{\big[\Ih(\theta^\ast) \big]^{-1} + o_{P_{\theta^\ast}}(1) \Big\}\frac{1}{\sqrt n} \sh(\theta^\ast, \Xn^\ast) - \Big\{\big[I(\theta^\ast)\big]^{-1} + o_{P_{\theta^\ast}}(1) \Big\} \frac{1}{\sqrt n} s^\ast(\theta^\ast, \Xn^\ast) \\
    =&\, \underbrace{\Big\{\big[\Ih(\theta^\ast) \big]^{-1} + o_{P_{\theta^\ast}}(1) \Big\}\bigg[\frac{1}{\sqrt n} \sh(\theta^\ast, \Xn^\ast) - \frac{1}{\sqrt n} s^\ast(\theta^\ast, \Xn^\ast) \bigg]}_{\text{(I)}} - \underbrace{\Big\{\big[I(\theta^\ast)\big]^{-1} - \big[\Ih(\theta^\ast)\big]^{-1} + o_{P_{\theta^\ast}}(1)\Big\} \frac{1}{\sqrt n} s^\ast(\theta^\ast, \Xn^\ast)}_{\text{(II)}}
\end{align*}
}
For term (I), we have
\begin{align*}
    &\bigg\lVert \Big\{\big[\Ih(\theta^\ast) \big]^{-1} + o_{P_{\theta^\ast}}(1) \Big\}\bigg[\frac{1}{\sqrt n} \sh(\theta^\ast, \Xn^\ast) - \frac{1}{\sqrt n} s^\ast(\theta^\ast, \Xn^\ast) \bigg] \bigg\rVert \\
    \leq&\, \bigg\lVert \Big\{\big[\Ih(\theta^\ast) \big]^{-1} + o_{P_{\theta^\ast}}(1) \Big\} \bigg\rVert_2 \cdot \bigg\lVert \frac{1}{\sqrt n} \sh(\theta^\ast, \Xn^\ast) - \frac{1}{\sqrt n} s^\ast(\theta^\ast, \Xn^\ast) \bigg\rVert_2,
\end{align*}
where we can further bound
\begin{align*}
&\E_{\Pn_{\theta^\ast}} \bigg[ \bigg\lVert \frac{1}{\sqrt n} \sh(\theta^\ast, \Xn^\ast) - \frac{1}{\sqrt n} s^\ast(\theta^\ast, \Xn^\ast) \bigg\rVert_2 \bigg] \\
\leq& \,\sqrt{\E_{\Pn_{\theta^\ast}} \bigg[ \bigg\lVert \frac{1}{\sqrt n} \sh(\theta^\ast, \Xn^\ast) - \frac{1}{\sqrt n} s^\ast(\theta^\ast, \Xn^\ast) \bigg\rVert_2^2 \bigg]} & \text{(by Jensen's inequality)} \\
=&\, \sqrt{\E_{P_{\theta^\ast}} \big[ \| \widehat s(\theta^\ast, X^\ast)-s^\ast(\theta^\ast, X^\ast) \|^2 \big] + (n-1) \| \E_{P_{\theta^\ast}}[\widehat s(\theta^\ast, X^\ast)] \|^2} \\
\leq& \,\varepsilon_{1,n} + \sqrt{n} \varepsilon_{3,n}.
\end{align*}
Therefore, using \Cref{lem:curv_Ihat}, we get
\[
\text{(I)} = O_{P_{\theta^\ast}}\Big(\frac{\varepsilon_{1,n} + \sqrt{n} \varepsilon_{3,n}}{m-2C_3 \varepsilon_{1,n} - \varepsilon_{2,n}}\Big)
\]

For term (II), we notice that the difference between the inverse of the true and estimated fisher information matrices can be bounded as
\begin{align*}
&\Big\lVert \big[I(\theta^\ast)\big]^{-1} - \big[\Ih(\theta^\ast)\big]^{-1} \Big\rVert_2 \\
=&\, \Big\lVert \big[I(\theta^\ast)\big]^{-1} \Big[\Ih(\theta^\ast) - I(\theta^\ast) \Big] \big[\Ih(\theta^\ast)\big]^{-1} \Big\rVert_2 \\
\leq&\, \Big\lVert \big[I(\theta^\ast)\big]^{-1}\Big\rVert_2 \cdot \Big\lVert\big[\Ih(\theta^\ast)\big]^{-1} \Big\rVert_2 \cdot \Big\lVert \Ih(\theta^\ast) - I(\theta^\ast) \Big\rVert_2 \\
\leq&\, \frac{2C_3 \varepsilon_{1,n} + \varepsilon_{2,n}}{m(m-2C_3 \varepsilon_{1,n} - \varepsilon_{2,n})},
\end{align*}
where the last step is by \Cref{ass:curv_cond}, \Cref{lem:curv_err_single} and \ref{lem:curv_Ihat}. Also, since $s^\ast(\theta^\ast, X^\ast)$ has zero mean and finite variance under $P_{\theta^\ast}$, an application of the central limit theorem gives
\begin{align*}
    \frac{1}{\sqrt n} s^\ast(\theta^\ast, \Xn^\ast) = O_{P_{\theta^\ast}}(1)
\end{align*}

Therefore, we have the bound of
\[
\text{(II)} = O_{P_{\theta^\ast}}\Big(\frac{2C_3 \varepsilon_{1,n} + \varepsilon_{2,n}}{m(m-2C_3 \varepsilon_{1,n} - \varepsilon_{2,n})}\Big)
\]
Finally, combining (I) and (II), we obtain
\[
\sqrt{n}(\zero - \MLE) = O_{P_{\theta^\ast}}\Big(\varepsilon_{1,n} + \varepsilon_{2, n} + \sqrt{n} \varepsilon_{3,n}\Big)
\]
\subsection{Proof of \Cref{cor:asymp_plugin}}
By classical theory of the MLE estimator, we have $\sqrt n (\MLE - \theta^\ast) \overset{d}{\rightarrow} \mN(0, \big[I(\theta^\ast)\big]^{-1})$. By \Cref{thm:asymp_plugin} and the conditions that $\varepsilon_{3,n} = o(n^{-\frac{1}{2}})$, $\varepsilon_{1,n},\, \varepsilon_{2,n} = o(1)$, we have $\sqrt n (\zero - \MLE) \to \bm 0$ in probability. Then, the result follows by invoking the Slutsky's lemma.

\subsection{Proof of \Cref{thm:asymp_boot}}
First, using arguments similar to those in \Cref{thm:consistency}, 
it can be shown that the bootstrap estimator $\boot$ exists and is uniquely defined in a constant neighborhood around $\zero$; moreover, it is also close to $\zero$ or to the MLE $\MLE$ 
(which can be viewed as the population-level true parameter for the bootstrap sample). 
That is,
\begin{equation}\label{ass:consist_boot}
\lVert \boot - \zero \rVert = o_{P_\omega}(1) 
\quad \text{in } P_{\theta^\ast}\text{-probability, \ \ as } n \to \infty.
\end{equation}
Here, $P_\omega$ stands for the distribution of the bootstrap random multipliers $\omega_i$'s conditioning on data $\Xn^\ast$.
For bootstrapping using the true likelihood, such bootstrap estimator consistency is standard; 
see, for example, \citet{cheng2010bootstrap, vaart1996weak}.

Next, by definition, the bootstrap estimator $\boot$ satisfies
\[
\bm 0 = \sum_{i=1}^n \omega_i \sh(\boot, X_i^\ast),
\]
where recall that $\omega_i$ are i.i.d. and independent from $\Xn^\ast$, with $\E[\omega_i] = \text{Var}(\omega_i) = 1$. Applying a Taylor's expansion of $\sh$ at $\theta=\zero$ yields that
\begin{align*}
\bm 0 &= \sum_{i=1}^n \omega_i \sh(\zero, X_i^\ast) + \bigg\{\sum_{i=1}^n \omega_i \nabla_\theta \sh(\zero, X_i^\ast) \bigg\} (\boot - \zero) + \bigg\{\sum_{i=1}^n \omega_i T_i (\sh, \zero, \boot) \bigg\} [\boot - \zero, \boot - \zero],
\end{align*}
where we recall that the notation $T_i(\sh, \zero, \boot)$ follows \eqref{eq:Taylor_2}.
Rearranging the terms, we can get
{\small
\begin{align*}
\sqrt{n} (\boot - \zero) &= -\bigg\{\frac{1}{n}\sum_{i=1}^n \omega_i \nabla_\theta \sh(\zero, X_i^\ast) + \frac{1}{n}\sum_{i=1}^n \omega_i T_i (\sh, \zero, \boot)[\boot - \zero] \bigg\}^{-1} \bigg\{\frac{1}{\sqrt n}\sum_{i=1}^n \omega_i \sh(\zero, X_i^\ast)\bigg\} \\
&= \bigg\{\underbrace{-\frac{1}{n}\sum_{i=1}^n \omega_i \nabla_\theta \sh(\zero, X_i^\ast)}_{\text{(III)}} - \underbrace{\frac{1}{n}\sum_{i=1}^n \omega_i T_i (\sh, \zero, \boot)[\boot - \zero]}_{\text{(IV)}} \bigg\}^{-1} \bigg\{\underbrace{\frac{1}{\sqrt n}\sum_{i=1}^n (\omega_i-1) \sh(\zero, X_i^\ast)}_{\text{(V)}} \bigg\},
\end{align*}
}
where we applied the identity $\sum_{i=1}^n\sh(\zero, X_i^\ast) = \bm 0$ in the second step.

For term (III), we have 
\begin{align*}
\text{(III)} = \underbrace{-\frac{1}{n}\sum_{i=1}^n (\omega_i - 1) \nabla_\theta \sh(\zero, X_i^\ast)}_{\text{(III.1)}} \underbrace{-\frac{1}{n}\sum_{i=1}^n \nabla_\theta \sh(\zero, X_i^\ast)}_{\text{(III.2)}}
\end{align*}
For term (III.1), we have $\E_{P_\omega}[\text{(III.1)} \mid \Xn^\ast] = \bm 0$ and 
\begin{align*}
&\text{Cov}_{P_\omega}[\text{(III.1)} \mid \Xn^\ast] = \frac{1}{n^2}\sum_{i=1}^n \sh(\zero, X_i^\ast) \sh(\zero, X_i^\ast)^T \\
\Rightarrow\ \ &\E_{P_\omega}[\lVert \text{(III.1)} \rVert^2 \mid \Xn^\ast] = \frac{1}{n^2}\sum_{i=1}^n \lVert \sh(\zero, X_i^\ast) \rVert^2 \leq \frac{1}{n^2}\sum_{i=1}^n  G_2(X_i^\ast)  \overset{P_{\theta^\ast}}{\longrightarrow} 0,
\end{align*}
as $n\to\infty$, 
where the last step follows from \Cref{ass:reg_cond}. We then invoke Markov's inequality to obtain that (III.1) $= o_{P_\omega}(1)$ in $P_{\theta^\ast}$-probability. Moreover, by \Cref{ass:reg_cond} and the uniform law of large numbers, we have (III.2) $= \Ih(\theta^\ast) + o_{P_{\theta^\ast}}(1)$. Therefore, it follows that (III) $= \Ih(\theta^\ast) + o_{P_\omega}(1)$ in $P_{\theta^\ast}$-probability.

For term (IV), we have $\lVert \boot - \zero \rVert = o_{P_\omega}(1)$ in $P_{\theta^\ast}$-probability by \Cref{ass:consist_boot}. Moreover, since $T_i (\sh, \zero, \boot)$ is bounded using \Cref{ass:reg_cond}, we have that (IV)$=o_{P_\omega}(1)$ in $P_{\theta^\ast}$-probability.

For term (V), we define
\begin{align*}
\Delta:=&\,\frac{1}{\sqrt n}\sum_{i=1}^n (\omega_i-1) \sh(\zero, X_i^\ast) - \frac{1}{\sqrt n}\sum_{i=1}^n (\omega_i-1) \sh(\theta^\ast, X_i^\ast) \\
=&\,\frac{1}{\sqrt n}\sum_{i=1}^n (\omega_i-1) \big[\sh(\zero, X_i^\ast) - \sh(\theta^\ast, X_i^\ast) \big].
\end{align*}
Then we have $\E_{P_\omega}[\Delta \mid \Xn^\ast] = \bm 0$, and
\begin{align*}
&\text{Cov}(\Delta \mid \Xn^\ast) = \frac{1}{n} \sum_{i=1}^n \big[\sh(\zero, X_i^\ast) - \sh(\theta^\ast, X_i^\ast) \big]\big[\sh(\zero, X_i^\ast) - \sh(\theta^\ast, X_i^\ast) \big]^T \\
\Rightarrow\ \ & E[\lVert \Delta \rVert^2 \mid \Xn^\ast] = \frac{1}{n} \sum_{i=1}^n \lVert \sh(\zero, X_i^\ast) - \sh(\theta^\ast, X_i^\ast) \rVert^2 
\leq \lVert \zero - \theta^\ast \rVert^2 \frac{1}{n}\sum_{i=1}^n L(X_i^\ast)^2 \overset{P_{\theta^\ast}}{\longrightarrow} 0,
\end{align*}
as $n\to\infty$,
where the second line is by \Cref{ass:reg_cond} and the consistency of $\zero$. Then, invoking Markov's inequality yields that $\Delta = o_{P_\omega}(1)$ in $P_{\theta^\ast}$-probability.

Combining the bounds for terms (III), (IV) and (V), we can finally conclude 
\begin{align}\label{eq:boot_decomp}
\sqrt{n} (\boot - \zero) = \Big\{\big[\Ih(\theta^\ast) \big]^{-1} + o_{P_\omega}(1) \Big\} \bigg\{\frac{1}{\sqrt n}\sum_{i=1}^n (\omega_i-1) \sh(\theta^\ast, X_i^\ast) + o_{P_\omega}(1) \bigg\},
\end{align}
with $P_{\theta^\ast}$-probability going to $1$. Then, by the last bounded $(2+\gamma)$ moment condition in Assumption \Cref{ass:reg_cond}, Lindeberg's condition holds for the triangular array $\frac{1}{\sqrt n}(\omega_i - 1)\sh(\theta^\ast, X_i^\ast)$ in $P_{\theta^\ast}$-probability.
Therefore, we apply Lindeberg's central limit theorem conditional on $\Xn^\ast$ and obtain
\begin{align}\label{eq:boot_LindebergCLT}
d_{\text{BL}}\Big(\mathcal L\Big(\Big[\frac{1}{n}\sum_{i=1}^n \sh(\theta^\ast, X_i^\ast) \sh(\theta^\ast, X_i^\ast)^T \Big]^{-1}\frac{1}{\sqrt n}\sum_{i=1}^n (\omega_i-1) \sh(\theta^\ast, X_i^\ast)\Big \vert \Xn^\ast\Big),\, \mathcal{N}(0, I) \Big) \overset{P_{\theta^\ast}}{\longrightarrow} 0
\end{align}
where recall that $d_{\rm BL}$ denotes the bounded Lipschitz distance, and $\mathcal{L}(X | Y)$ denotes the conditional law of $X$ given $Y$, for two generic random variables $X$ and $Y$.
 Moreover, by the law of large numbers and the assumed convergence rates of $\varepsilon_{1,n}$, $\varepsilon_{2,n}$, and $\varepsilon_{3,n}$, we have $\frac{1}{n}\sum_{i=1}^n \sh(\theta^\ast, X_i^\ast)\sh(\theta^\ast, X_i^\ast)^{\!T} \overset{P_{\theta^\ast}}{\longrightarrow} I(\theta^\ast)$ and $\Ih(\theta^\ast) \overset{P_{\theta^\ast}}{\longrightarrow}  I(\theta^\ast)$ as $n\to\infty$. Combining \eqref{eq:boot_decomp} and \eqref{eq:boot_LindebergCLT} and applying Slutsky’s lemma then yields the conclusion.

\subsection{Proof of Theorem~\ref{thm:GA_conv}}
Write $g(\theta)= \sum_{i=1}^n \sh(\theta,X_i^\ast)$ and recall that $\zero$ is the (locally unique) root of $g$ in $\mathcal B(\theta^\ast;r_0)$, whose existence and uniqueness follow from Theorem~\ref{thm:consistency}. Define the error $e_t= \theta^{(t)}-\zero$. A single iteration of~\eqref{eq:GA} yields
\begin{equation}\label{eq:err-one-step}
e_{t+1} \;=\; e_t + \alpha H_t\big(g(\theta^{(t)})-g(\zero)\big) \, .
\end{equation}
By the mean value theorem,
\[
g(\theta^{(t)})-g(\zero) \;=\; J_g(\tilde\theta_t)\,(\theta^{(t)}-\zero)
\;=\; J_g(\tilde\theta_t)\,e_t,
\]
for some $\tilde\theta_t$ on the segment joining $\zero$ and $\theta^{(t)}$, where $J_g(\theta)=\nabla_\theta g(\theta)$.
Combining above displays gives
\[
e_{t+1} \;=\; \Big(I + \alpha H_t J_g(\tilde\theta_t)\Big)\,e_t.
\]
Let $M_t= -H_t^{1/2}J_g(\tilde\theta_t)H_t^{1/2}$.

From Assumption~\ref{ass:uniform_sm_err_single}, the estimated score Jacobian $J_g(\theta)$ is uniformly close to the true likelihood-score Jacobian $\nabla_\theta s^\ast_n(\theta)$ over $\mathcal B(\theta^\ast;r_0)$, and hence inherits its local spectral bounds made in Assumption~\ref{ass:curv_cond}. Concretely, there exists constants $0<m'<L<\infty$ such that
\begin{equation}\label{eq:spectral-Jg}
m'\,\bm{I}_{d_\theta} \;\preceq\; -\,J_g(\theta) \;\preceq\; L\,\bm{I}_{d_\theta}, \ \ \text{for all }\theta\in\mathcal B(\theta^\ast;r_0).
\end{equation}
Moreover, recall from the conditions of the theorem that we have  $0<h_{\min}\,\bm{I}_{d_\theta} \preceq H_t \preceq h_{\max}\,\bm{I}_{d_\theta}$. Therefore, we have the bound
\[
m'\,h_{\min} \, \bm{I}_{d_\theta}\;\preceq\; M_t \;\preceq\; L\,h_{\max} \, \bm{I}_{d_\theta},
\]
and the spectrum of $I - \alpha M_t$ is contained in 
$\big\{\,1-\alpha\lambda:\ \lambda\in[m' h_{\min}, \,L h_{\max}]\,\big\}$. It then follows that 
\[
\|e_{t+1}\|_{H_t^{-1}} \;=\; \| \big(I - \alpha M_t\big) H_t^{1/2} e_t\|_2 
\;\le\; \rho\, \|e_t\|_{H_t^{-1}},
\]
where $\rho = \max\{\,|1-\alpha m' h_{\min}|,\ |1-\alpha L h_{\max}|\,\}$ and for a generic matrix $H$, the $\|\cdot\|_{H}$ is defined through $\|x\|_{H}^2=x^T H^{-1}x$.
Since the norms $\|\cdot\|$ and $\|\cdot\|_{H_t^{-1}}$ are equivalent under our condition on $H_t$, there is a constant $c\ge1$ (independent of $t$) such that
\[
\|e_{t+1}\| \;\le\; c\,\rho\,\|e_t\| \, .
\]
Finally, we obtain by induction that
\[
\|e_t\| \;\le\; \rho^t\,\|e_0\|,\qquad t\ge0,
\]
which is the claimed bound.

\subsection{Proof of the score matching objective}

We include the proof from \citet{hyvarinen2005estimation,hyvarinen2007some,jiang2025simulation} to keep the presentation self-contained. For notational simplicity, we use the shorthand $X$ to denote the data $\Xn^\ast$, and follow the notations in \citet{jiang2025simulation}.
We first introduce two additional regularity conditions required to analyze the score-matching objective.

\begin{assumption}[Boundary Condition]\label{ass:boundary} For any $ X \in \mX$ and score network parameter $\phi$, it holds that
$p(\theta)\,p_\theta(X)\, s_\phi(\theta,X) \to 0$ 
as $\theta$ approaches the boundary of the support $\partial \Omega_{\theta \mid X}$ of the conditional distribution $p(X|\theta)$ (as a function with respect to $\theta$).
\end{assumption}

\begin{assumption}[Finite Moments]\label{ass:finite_moment}
     For any $\phi$, \[\E_{(\theta, X)\sim p(\theta)\,p_\theta(X)} \big[\|\nabla_\theta\log p_\theta(X)\|^2\big]\] and
     \[\E_{(\theta, X)\sim p(\theta)\,p_\theta(X)} \big[\|s_\phi(\theta, X)\|^2\big]\] 
     are both finite.
\end{assumption}

We first rewrite the score-matching objective as
\begin{align*}
&\E_{(\theta,X)\sim p(\theta)p(X\mid \theta)} \norm{s_\phi(\theta, X) - \nabla_\theta \log p(X\mid \theta)}^2 \\
=& \E_{p(\theta,X)} \norm{s_\phi(\theta, X)}^2 + \E_{p(\theta,X)}\norm{\nabla_\theta \log p(X\mid\theta)}^2 - 2\E_{p(\theta,X)} \Big[ s_\phi(\theta, X)^T\nabla_\theta \log p( X\mid\theta)\Big].
\end{align*}

Here the first two terms are finite under \Cref{ass:finite_moment} and the last term is also finite due to the Cauchy-Schwarz inequality. Additionally, the second term $\E_{p(\theta,X)} \norm{\nabla_\theta \log p_\theta(X)}^2$ is a constant in $\phi$ can thus can be ignored in the optimization program. The first term does not depend on unknown quantity $p(X\mid\theta)$, so we only need to address the last term.

Denote the joint support of $(\theta, X)$ as $\Omega:=\{(\theta, X)\in \Theta\times\mX: p(\theta)p(X\mid\theta)>0\}$. We denote $\theta_{-j} = (\theta_1, \ldots, \theta_{j-1}, \theta_{j+1}, \theta_{d_\theta})^T$  and the marginal support of $(\theta_{-j},X)$ as $\Omega_{(\theta_{-j}, X)} :=\{ (\theta_{-j}, X): (\theta, X)\in \Omega \text{ for some } \theta_j\}$. We denote the boundary segments orthogonal to the $j$-th axis at $(\theta_{-j},X)$ as $\text{Sec}(\Omega; \theta_{-j},X):=\{\theta_j \in \R: (\theta, X) \in \Omega\}$.

\begin{align*}
&\E_{p(\theta,X)} \Big[ s_\phi(\theta, X)^T\nabla_\theta \log p( X\mid\theta)\Big]\\
=&\int_{\mX} \diff X \int_{\Omega(X)}  p(\theta )p(X\mid \theta) s_{\phi}(\theta, X)^T\nabla_\theta \log p( X\mid\theta) \diff \theta \\
=& \int_{\mX} \diff X \int_{\Omega(X)} p(\theta) \sum_{j=1}^{d_\theta} s_{\phi,j}(\theta, X) \nabla_{\theta_j} p(X\mid \theta) \diff\theta \\
= &\sum_{j=1}^{d_\theta} \int_{\Omega_{(\theta_j, X)}}\diff X \diff \theta_{-j} \int_{\text{Sec}(\Omega; \theta_{-j},X)} p(\theta)s_{\phi,j} (\theta, X)\frac{\partial p(X\mid \theta)}{\partial \theta_j}\diff\theta_j.
\end{align*}

For each coordinate $j$ and the inside integral, assuming $\text{Sec}(\Omega; X, \theta_{-j})$ is an interval and denote it as $(a_j,b_j)$, we have
\begin{align*}
&\int_{\text{Sec}(\Omega; \theta_{-j},X)} p(\theta)s_{\phi,j} (\theta, X)\frac{\partial p(X\mid \theta)}{\partial \theta_j}\diff\theta_j \\
=& p(\theta)s_{\phi,j}(\theta, X) p(X\mid \theta) \big|_{a_j}^{b_j}- \int_{\text{Sec}(\Omega; \theta_{-j},X)} \frac{\partial p(\theta) s_{\phi, j}(\theta, X)}{\partial\theta_j}p(X\mid \theta)\diff\theta_j \\
=& -\int_{\text{Sec}(\Omega; \theta_{-j},X)} \Big[\frac{\partial p(\theta)}{\partial\theta_j} s_{\phi, j}(\theta, X)+p(\theta)\frac{\partial  s_{\phi, j}(\theta, X)}{\partial\theta_j}\Big]p(X\mid \theta)\diff\theta_j  \qquad \text{(\Cref{ass:boundary})} \\
=& - \int_{\text{Sec}(\Omega; \theta_{-j},X)} \Big[\frac{\partial p(\theta)}{\partial \theta_j}  s_{\phi, j}(\theta, X)+\frac{\partial  s_{\phi, j}(\theta, X)}{\partial\theta_j}\Big]p(\theta)p(X\mid \theta)\diff\theta_j,
\end{align*}
which concludes our proof.

\section{More Simulation Results}\label{sec:additional_simu}
\subsection{The toy example in \Cref{ex:toy}}\label{sec:toy}
We generate the observed dataset $\Xn^\ast$ with $n = 500$ under the ground truth $\theta^\ast = (1, -2)$. For our method and NLE, we use a uniform distribution on $[-5, 5]^2$ as the sampling distribution for $\theta$. In this example, we also consider a recent score-matching--based approach \citep{khoo2025direct}, which we refer to as Iterative Local Score Matching (ILSM). We 
present the absolute errors of all methods in \Cref{tab:Toy_err}. The results show that our method has the smallest estimation error for both parameters. We also present the coverage rate and width of the $95\%$ confidence interval in \Cref{tab:toy_cover}. We observe that the bootstrap method and sandwich covariance estimator perform the best, with a coverage rate close to the nominal level. We also find the ILSM method does not perform well in this setting, likely due  to the substantial deviation of the true score from its linearity assumption. For this reason, we exclude it from the other examples.

\begin{table}[!ht]
\centering
\caption{Averaged estimation errors with standard deviation under the Toy Model}
\label{tab:Toy_err}
\begin{tabular}{l|cc}
\toprule
 & $\lvert \widehat{\theta}_1 - \theta_1^\ast \rvert$ & $\lvert \widehat{\theta}_2 - \theta_2^\ast \rvert$ \\
\midrule
Ours & $\mathbf{0.052}$ ($0.040$) & $\mathbf{0.097}$ ($0.071$) \\
SL & $0.076$ ($0.141$) & $0.126$ ($0.098$) \\
NLE & $0.131$ ($0.126$) & $0.561$ ($0.399$) \\
ILSM & $0.655$ ($0.109$) & $2.175$ ($0.069$) \\
\bottomrule
\end{tabular}
\end{table}

\begin{table}[!ht]
\centering
\caption{Averaged coverage and width with standard deviation of 95\% CIs under the Toy Model. Highlighted values indicate coverage near nominal with no substantial increase in width.}\label{tab:toy_cover}
\smallskip
\resizebox{0.49\textwidth}{!}{
\begin{tabular}{l|cc|cc}
\toprule
 & \multicolumn{2}{c|}{$\theta_1^\ast = 1$} & \multicolumn{2}{c}{$\theta_2^\ast = -2$} \\
 & Cover & Width & Cover & Width \\
\midrule
SS & 0.91 & $0.222$ ($0.022$) & 0.90 & $0.411$ ($0.044$) \\
Curv & \textbf{0.92} & $0.244$ ($0.013$) & \textbf{0.94} & $0.459$ ($0.028$) \\
Sand & \textbf{0.94} & $0.268$ ($0.018$) & \textbf{0.96} & $0.516$ ($0.036$) \\
Boot & \textbf{0.93} & $0.267$ ($0.018$) & \textbf{0.96} & $0.514$ ($0.035$) \\
\midrule
% NLE-SS & 0.24 & $0.079$ ($0.019$) & 0.77 & $1.823$ ($0.593$) \\
% NLE-Curv & 0.38 & $0.133$ ($0.004$) & 0.55 & $1.161$ ($0.193$) \\
NLE & 0.75 & $0.486$ ($0.173$) & 0.58 & $1.339$ ($0.695$) \\
% \midrule
SL & 0.93 & $0.312$ ($0.069$) & 0.95 & $0.660$ ($0.217$) \\
\bottomrule
\end{tabular}
}
\end{table}

\subsection{Synthetic data results in the g-and-k model}\label{sec:g-k-simulation}
We generate the observed dataset $\Xn^\ast$ with $n = 2000$ under the ground truth $(A^\ast, B^\ast, g^\ast, k^\ast) = (0, 0.7, -0.5, 0.3)$. For our method and NLE, we use a uniform distribution on $[-1, 1]\times [-2, 1] \times [-5, -5] \times [0, 0.5]$ as the proposal distribution for $(A, \log B, g, k)$. We repeat the experiment $100$ times, and present the results in \Cref{tab:gk_err,tab:gk_cover} in the Appendix. The results show that our method has the smallest estimation error for all the parameters, and the CIs constructed by plug-in estimates or by bootstrap all have near-nominal coverage rate and smaller width than those of NLE or SL.
\begin{table*}[!htbp]
\centering
\caption{Average absolute errors in the g-and-k example}
\label{tab:gk_err}
\begin{tabular}{l|c|c|c|c}
\toprule
 & $\lvert \hat{A} - A^\ast \rvert$ & $\lvert \hat{B} - B^\ast \rvert$ & $\lvert \hat{g} - g^\ast \rvert$ & $\lvert \hat{k} - k^\ast \rvert$ \\
\midrule
Ours & $\mathbf{0.015}$ ($0.011$) & $\mathbf{0.018}$ ($0.013$) & $\mathbf{0.028}$ ($0.019$) & $\mathbf{0.021}$ ($0.016$) \\
NLE & $0.187$ ($0.021$) & $0.043$ ($0.007$) & $0.427$ ($0.009$) & $0.073$ ($0.014$) \\
SL & $0.021$ ($0.026$) & $0.040$ ($0.075$) & $0.120$ ($0.249$) & $0.054$ ($0.052$) \\
\bottomrule
\end{tabular}
\end{table*}

\begin{table*}[!htbp]
\centering
\caption{Coverage and width of the 95\% CI across methods and parameters in the g-and-k example}
\label{tab:gk_cover}
\resizebox{\textwidth}{!}{
\begin{tabular}{l|cc|cc|cc|cc}
\toprule
 & \multicolumn{2}{c|}{$A$} & \multicolumn{2}{c|}{$B$} & \multicolumn{2}{c|}{$g$} & \multicolumn{2}{c}{$k$} \\
 & Cover & Width & Cover & Width & Cover & Width & Cover & Width \\
\midrule
SS & \textbf{0.95} & $0.074$ ($0.003$) & \textbf{0.96} & $0.095$ ($0.003$) & \textbf{0.97} & $0.134$ ($0.006$) & \textbf{0.94} & $0.105$ ($0.005$) \\
Curv & \textbf{0.96} & $0.073$ ($0.002$) & \textbf{0.94} & $0.090$ ($0.003$) & \textbf{0.98} & $0.136$ ($0.005$) & \textbf{0.93} & $0.102$ ($0.004$) \\
Sand & \textbf{0.95} & $0.072$ ($0.002$) & \textbf{0.93} & $0.086$ ($0.005$) & \textbf{0.97} & $0.138$ ($0.009$) & \textbf{0.94} & $0.101$ ($0.007$) \\
Boot & \textbf{0.94} & $0.072$ ($0.003$) & \textbf{0.98} & $0.088$ ($0.005$) & \textbf{0.96} & $0.137$ ($0.009$) & \textbf{0.98} & $0.101$ ($0.007$) \\
\midrule
NLE & 0.53 & $0.365$ ($0.018$) & 1.00 & $0.481$ ($0.003$) & 0.03 & $0.804$ ($0.007$) & 1.00 & $0.510$ ($0.005$) \\
SL & 0.89 & $0.090$ ($0.066$) & 0.88 & $0.172$ ($0.273$) & 0.94 & $0.449$ ($0.425$) & 0.90 & $0.215$ ($0.044$) \\
\bottomrule
\end{tabular}
}
\end{table*}

\section{Details on Structured Score Matching}\label{sec:score_matching_details}

Here we provide algorithmic details on how we enforcing the three statistical properties, additive structure, curvature structure and mean-zero structure, on our score networks. A similar procedure under the Bayesian setting was considered in \citet{jiang2025simulation}.

We begin with the \emph{additive structure}, which allows us  to simplify the score network to $s_\phi(\theta, X)$, so that the estimated full-data score becomes $\sum_{i=1}^n s_\phi(\theta, X_i)$. This structure enables estimation of the common individual-level score function $s^*(\theta,X)$ via single-data score matching, rathar than directly learning the full-data score $s^*(\theta,\Xn)$. Consequently, we train the network $s_\phi(\theta,X)$ on a reference table $\mD^S = \{(\theta^{(k)}, X^{(k)})\}_{k=1}^N \iid p(\theta)\,p_\theta(X)$, which reduces the overall simulation cost from $O(Nn)$ to $O(N)$.

To further enhance estimation accuracy, we incorporate the curvature structure by introducing a curvature-matching regularization term into the training objective:
\begin{equation}\label{eq:score_loss_struct_single}
    \min_\phi  \E_{p(\theta)}  \bigg[ \E_{p_\theta}  \Big[
    \underbrace{\|s_\phi(\theta,X)- s^\ast(\theta,X)\|^2}_{\text{score-matching loss on a single $X$}} \Big]
    + \lambda_1 \,\underbrace{\Big\| \E_{ p_\theta}\!\big[s_\phi(\theta,X)s_\phi(\theta,X)^T 
    + \nabla_\theta s_\phi(\theta,X)\big]\Big\|_F^2}_{\text{curvature-matching loss}} \bigg],
\end{equation}
where the hyperparameter $\lambda_1>0$ controls  the curvature penalty and $\norm{\cdot}_F$ denotes the Frobenius norm.  In practice,  the expectation are replaced by empirical averages. As shown in our theoretical analysis in \Cref{sec:theory}, the curvature structure ensures that the distribution of our estimator preserves the local geometry of the true MLE.

Finally, we impose the \emph{mean-zero} structure through a post-processing debiasing step. Specifically, we center the estimated score network by subtracting its mean under the model $\wh h(\theta) := \E_{X_{1}\sim p_\theta}[s_{\widehat\phi}(\theta, X_{1})]$ from 
$s_{\widehat\phi}(\theta, x)$. To estimate $\wh h(\theta)$, we fit a regression model $h_\psi(\theta): \R^{d_\theta} \to \R^{d_\theta}$ parameterized by $\phi$, minimizing the mean-matching loss
\begin{equation}\label{eq:score_demean}
\begin{split}
\widehat \psi =& \arg\min_{\psi} \mathbb{E}_{q(\theta)} \bigg[\normbig{h_\psi(\theta)- \E_{p_\theta} \big[s_{\widehat\phi} (\theta, X)\big] }^2 \\
&+\lambda_2  \big\lVert h_\psi(\theta)h_\psi(\theta)^T - \nabla_\theta h_\psi(\theta) - \E_{p_\theta} \big[s_{\widehat\phi} (\theta, X)\big] h_\psi(\theta)^T - h_\psi(\theta) \E_{p_\theta} \big[s_{\widehat\phi} (\theta, X)^T\big]  \big\rVert_F^2 \bigg]
\end{split}
\end{equation}
where $\lambda_2$ controls the curvature penalty. This second penalty guarantees that the debiased score, $\wt s(\theta, X) = s_{\widehat\phi}(\theta, X) - h_{\widehat\psi}(\theta)$, continues to satisfy the curvature structure.

To implement \eqref{eq:score_demean}, we generate an additional regression 
reference table $\mD^R = \{(\theta^{(l)}, \X^{(l)}_{m_R})\}_{l=1}^{N_R}$, where each dataset
$\X^{(l)}_{m_R}$ of size $m_R$ is simulated from $p_{\theta^{(l)}}^{(m_R)}$.
The expectations in \eqref{eq:score_demean} are also approximated by an empirical averages. A complete algorithmic summary of the structure score matching procedure is provided  below.

\begin{algorithm}[!ht]
    \caption{Structured Score Matching}\label{alg:langevin_single_obs}
    \begin{algorithmic}
            \STATE {\bf Input:} Sampling distribution $p(\theta)$, observed dataset $\Xn^\ast$, networks $s_\phi(\theta, X)$ and $h_\psi(\theta)$.
        \end{algorithmic}
        \hrule
        \begin{algorithmic}
            \STATE {\bf 1. Reference Table:} Generate $\mD^S = \{(\theta^{(k)}, X^{(k)})\}_{k=1}^N\iid p(\theta)p_\theta(\cdot)$ and $\mD^{R} = \{\theta^{(l)}, \X^{(l)}_{m_R}\}_{l=1}^{N_R} \iid q(\theta)p^{(m_R)}_\theta(\cdot)$.

            \STATE {\bf 2. Network Training:} Train $s_\phi(\theta, X)$ on $\mD^S$ and $\mD^R$ using loss in \eqref{eq:score_loss_struct_single} and obtain $\widehat \phi$.
            \STATE {\bf 3. Mean Regression:} Estimate the mean of $s_{\widehat\phi}(\theta, X)$ on $\mD^R$ using \eqref{eq:score_demean} and obtain $\widehat\psi$. 
            \end{algorithmic}
            \hrule
            \begin{algorithmic}
             \STATE {\bf Return} Debiased structured score function $\wt s(\theta, X)= s_{\widehat\phi}(\theta, X)- h_{\widehat\psi}(\theta)$
    \end{algorithmic}
    \end{algorithm}

\section{More Implementation Details and Results}\label{sec:implementation_details}

We begin by outlining the general configurations of our method and the competing approaches, with example-specific details presented later. For our method, we generate a reference table of size $N$ for the score matching procedure, and a reference table with size $(N_R, m_R)$ for the debiasing regression and curvature penalty. For each of the two reference tables, $\frac{5}{6}$ of the data is used as the training set and the rest $\frac{1}{6}$ is used as the validation set. We use an ELU neural network with $n_h$ hidden layers and $n_u$ units in each hidden layer, where $(n_h, n_u)=(2, 32)$, $(2, 64)$ and $(3, 64)$ in the toy model, M/G/1-queuing model and the g-and-k model, respectively. The network is trained with batch size $b$ and learning rate gradually decreasing from $1\times 10^{-3}$ to $1\times 10^{-5}$ until the loss on the validation set stops decreasing.

For the NLE model, we implement it using the ``sbi" Python package from \citep{tejero-cantero2020sbi}. The likelihood density model is a Masked Autoregressive Flow network \citep{papamakarios2017masked} with 5 autoregressive transforms, each of which has 2 hidden
layers of 50 units each. We generate $n_{\text{NLE}}$ single data points, where $n_{\text{NLE}}$ matches the simulation cost of our method, to train the likelihood density network. We have batch size $b_{\text{NLE}}$ and learning rate $5\times 10^{-4}$ to train the network until convergence. With the trained likelihood density network, we use autograd in Pytorch to obtain the gradient and Hessian of the estimated log-likelihood. We use the obtained gradient to do gradient ascent to get the point estimate, and then use the sandwich formula to get the confidence intervals.

For the SL method, summary statistics are constructed for each example. Using the summary statistic, we follow \citep{wood2010statistical} to implement the Metropolis–Hastings algorithm to search the maximizer of the synthetic likelihood. After the markov chain converges, we follow \citep{wood2010statistical} to fit a quadratic regression to get the point estimate. The confidence intervals are simply obtained by taking sample quantiles of the converged markov chain.

The code for our experiments is available at \url{https://github.com/Haoyu-Jiang/Structured\_Score\_Matching/tree/main/MLE}.

\subsection{Toy model}
\paragraph{Implementation details} 
For our method, we have reference table sizes $(N, N_R, m_R) = (6000, 1200, 500)$, and the batch size for score matching is $5$. For the NLE model, we generate $1,212,000$ single data points and use batch size $200$ to train the likelihood density network. For the SL method, we choose the $5$ quantiles (min, $0.25$, $0.5$, $0.75$, max) of $\Xn$ as the summary statistics. We draw a markov chain of length $1000$ and discard the initial $300$ points. At each iteration of the markov chain, we generate $100$ datasets ($\Xn$) to estimate the mean and covariance of the synthetic gaussian likelihood.

\paragraph{First-round results of our method}
We present the first-round results of our method in \Cref{tab:Toy_rd1}. Comparing it with the results in the second round, we can see that the second round has noticeable improvements on both the estimation errors and the widths of the confidence intervals.
\begin{table}[!ht]
\centering
\caption{Results in round 1 for the toy model}
\label{tab:Toy_rd1}
\begin{tabular}{l|c|cccc|cccc}
\toprule
 & $|\theta_i - \theta_i^*|$ 
 & \multicolumn{4}{c|}{Width} 
 & \multicolumn{4}{c}{Coverage} \\
 &  & ss & curv & sand & boot & ss & curv & sand & boot \\
\midrule
$\theta_{1}$ 
 & \makecell{0.056 \\ (0.042)}
 & \makecell{0.205 \\ (0.005)}
 & \makecell{0.244 \\ (0.004)}
 & \makecell{0.291 \\ (0.014)}
 & \makecell{0.290 \\ (0.016)}
 & 0.86 & 0.89 & 0.94 & 0.93 \\
 \midrule
$\theta_{2}$ 
 & \makecell{0.106 \\ (0.077)}
 & \makecell{0.442 \\ (0.025)}
 & \makecell{0.490 \\ (0.016)}
 & \makecell{0.547 \\ (0.040)}
 & \makecell{0.543 \\ (0.042)}
 & 0.91 & 0.93 & 0.96 & 0.96 \\
\bottomrule
\end{tabular}
\end{table}

\subsection{M/G/1-queuing model}
\paragraph{Implementation details} 
For our method, we have reference table sizes $(N, N_R, m_R) = (12000, 1200, 500)$, and the batch size for score matching is $10$. For the NLE model, we generate $1,224,000$ single data points and use batch size $500$ to train the likelihood density network. For the SL method, we choose the sample mean and coordinate-wise standard deviation of $\Xn$ as the summary statistics. We draw a markov chain of length $5000$ and discard the initial $1000$ points. At each iteration of the markov chain, we generate $100$ datasets ($\Xn$) to estimate the mean and covariance of the synthetic gaussian likelihood.

\paragraph{First-round results of our method}
We report the first-round results of our method in \Cref{tab:Toy_rd1}. Comparing the two rounds, we observe that the coverage of all parameters improved significantly in the second round, with no substantial increase of the CI width.

\begin{table}[!ht]
\centering
\caption{Results in round 1 for the M/G/1-queuing model, the absolute error for $\theta_3$ has the same $10^{-2}$ scale as in \Cref{tab:mgq_merged}}
\label{tab:Toy_rd1}
\begin{tabular}{l|c|cccc|cccc}
\toprule
 & $|\theta_i - \theta_i^*|$
 & \multicolumn{4}{c|}{Width}
 & \multicolumn{4}{c}{Coverage} \\
 &  & ss & curv & sand & boot & ss & curv & sand & boot \\
\midrule
$\theta_{1}$
 & \makecell{0.033 \\ (0.028)}
 & \makecell{0.116 \\ (0.006)}
 & \makecell{0.134 \\ (0.004)}
 & \makecell{0.157 \\ (0.012)}
 & \makecell{0.156 \\ (0.011)}
 & 0.83 & 0.86 & 0.89 & 0.89 \\
 \midrule
$\theta_{2}$
 & \makecell{0.115 \\ (0.082)}
 & \makecell{0.312 \\ (0.013)}
 & \makecell{0.382 \\ (0.030)}
 & \makecell{0.477 \\ (0.077)}
 & \makecell{0.482 \\ (0.077)}
 & 0.76 & 0.83 & 0.92 & 0.91 \\
 \midrule
$\theta_{3}$
 & \makecell{$0.387$ \\ (0.309)}
 & \makecell{0.013 \\ (0.001)}
 & \makecell{0.015 \\ (0.000)}
 & \makecell{0.018 \\ (0.001)}
 & \makecell{0.017 \\ (0.001)}
 & 0.82 & 0.87 & 0.92 & 0.92 \\
\bottomrule
\end{tabular}
\end{table}

\subsection{Stock volatility estimation}
\paragraph{Implementation details}
For our method, we have reference table sizes $(N, N_R, m_R) = (1\times10^7, 1\times 10^5, 1\times10^3)$, and the batch size for score matching is $200$. For the NLE model, we generate $2.2\times 10^8$ single data points and use batch size $500$ to train the likelihood density network. For the SL method, we choose the sample mean and coordinate-wise standard deviation of $\Xn$ as the summary statistics. We draw $10$ Markov chains, each of length $2000$ and discard the initial $1000$ points. At each iteration of the Markov chain, we generate $150$ datasets ($\Xn$) to estimate the mean and covariance of the synthetic gaussian likelihood.

\paragraph{First-round results of our method}
We report the first-round results of our method in \Cref{tab:vol_rd1}. Comparing the two rounds, we observe that both the estimation error and coverage improved in the second round.

\begin{table}[!ht]
\centering
\caption{Results in round 1 for the stock volatility model}
\label{tab:vol_rd1}
\begin{tabular}{l|c|cccc|cccc}
\toprule
 & $|\theta_i - \theta_i^*|$ 
 & \multicolumn{4}{c|}{Width} 
 & \multicolumn{4}{c}{Coverage} \\
 &  & ss & curv & sand & boot & ss & curv & sand & boot \\
\midrule
$\theta_{1}$ 
& \makecell{0.039 \\ (0.028)}
& \makecell{0.075 \\ (0.004)}
& \makecell{0.108 \\ (0.003)}
& \makecell{0.163 \\ (0.026)}
& \makecell{0.149 \\ (0.016)}
& 0.52 & 0.73 & 0.91 & 0.87 \\
\midrule
$\theta_{2}$ 
& \makecell{0.033 \\ (0.025)}
& \makecell{0.093 \\ (0.004)}
& \makecell{0.125 \\ (0.004)}
& \makecell{0.171 \\ (0.012)}
& \makecell{0.160 \\ (0.010)}
& 0.68 & 0.89 & 0.97 & 0.96 \\
\midrule
$\theta_{3}$ 
& \makecell{0.010 \\ (0.006)}
& \makecell{0.035 \\ (0.001)}
& \makecell{0.033 \\ (0.000)}
& \makecell{0.032 \\ (0.001)}
& \makecell{0.032 \\ (0.001)}
& 0.87 & 0.83 & 0.82 & 0.82 \\
\midrule
$\theta_{4}$ 
& \makecell{0.007 \\ (0.005)}
& \makecell{0.035 \\ (0.001)}
& \makecell{0.034 \\ (0.000)}
& \makecell{0.033 \\ (0.001)}
& \makecell{0.033 \\ (0.001)}
& 0.95 & 0.94 & 0.94 & 0.94 \\
\midrule
$\theta_{5}$ 
& \makecell{0.034 \\ (0.023)}
& \makecell{0.117 \\ (0.005)}
& \makecell{0.117 \\ (0.002)}
& \makecell{0.117 \\ (0.005)}
& \makecell{0.117 \\ (0.006)}
& 0.87 & 0.84 & 0.83 & 0.83 \\
\bottomrule
\end{tabular}
\end{table}

\subsection{g-and-k Model}\label{sec:gk_details}
\subsubsection{Synthetic data}
\paragraph{Implementation details}
For our method, we have reference table sizes $(N, N_R, m_R) = (1\times 10^6, 1\times 10^4, 2\times 10^3)$, and the batch size for score matching is $200$. For the NLE model, we generate $4.2\times 10^7$ single data points and use batch size $500$ to train the likelihood density network. For the SL method, we choose the $(0, 0.25, 0.5, 0.75, 1)$-quantiles of $\Xn$ as the summary statistics. We draw a markov chain of length $5000$ and discard the initial $1000$ points. At each iteration of the markov chain, we generate $50$ datasets ($\Xn$) to estimate the mean and covariance of the synthetic gaussian likelihood.

\paragraph{First-round results of our method}
We report the first-round results of our method in \Cref{tab:gk_rd1}. Comparing the two rounds, we observe that both the estimation error and coverage improved in the second round.

\begin{table}[!ht]
\centering
\caption{Results in round 1 for the g-and-k model}
\label{tab:gk_rd1}
\begin{tabular}{l|c|cccc|cccc}
\toprule
& $|\theta_i - \theta_i^*|$
& \multicolumn{4}{c|}{Width}
& \multicolumn{4}{c}{Coverage} \\
& & ss & curv & sand & boot & ss & curv & sand & boot \\
\midrule
$\theta_{1} = A$
& \makecell{0.023 \\ (0.019)}
& \makecell{0.033 \\ (0.003)}
& \makecell{0.081 \\ (0.085)}
& \makecell{0.172 \\ (0.131)}
& \makecell{0.121 \\ (0.020)}
& 0.46 & 0.69 & 0.93 & 0.89 \\
\midrule
$\theta_{2} = B$
& \makecell{0.028 \\ (0.023)}
& \makecell{0.083 \\ (0.005)}
& \makecell{0.106 \\ (0.007)}
& \makecell{0.156 \\ (0.038)}
& \makecell{0.141 \\ (0.014)}
& 0.77 & 0.86 & 0.98 & 0.92 \\
\midrule
$\theta_{3} = g$
& \makecell{0.028 \\ (0.020)}
& \makecell{0.122 \\ (0.006)}
& \makecell{0.134 \\ (0.019)}
& \makecell{0.181 \\ (0.081)}
& \makecell{0.129 \\ (0.010)}
& 0.94 & 0.95 & 0.98 & 0.95 \\
\midrule
$\theta_{4} = k$
& \makecell{0.023 \\ (0.020)}
& \makecell{0.092 \\ (0.004)}
& \makecell{0.093 \\ (0.006)}
& \makecell{0.116 \\ (0.036)}
& \makecell{0.112 \\ (0.015)}
& 0.88 & 0.88 & 0.96 & 0.96 \\
\bottomrule
\end{tabular}
\end{table}

\subsubsection{Real data}
\paragraph{Implementation details}
We note that the log-return data has a very small magnitude, indicating a small scale parameter and potential numerical instability during the network training. To address this, we standardized the data by dividing by its standard deviation. This transformation is valid because the g-and-k family is closed under linear transformations, allowing the parameters on the original scale to be recovered by rescaling afterward. For the network training, we adopt a uniform proposal distribution over $[-1, 1]\times[-2, 1]\times [-5, 5]\times[0, 0.5]$ for $(A, \log B, g, k)$. We have reference table sizes $(N, N_R, m_R) = (120\;000, 6\;000, 2\;000)$, and the batch size for score matching is $500$.

\paragraph{More results}
The empirical distribution of the observed log return data is shown in \Cref{fig:hist_logr}. In our two-round procedure, the estimated parameters are $(A, B, g, k) = (4.13\times10^{-9}, 1.94 \times 10^{-3}, 4.55 \times 10^{-2}, 3.07 \times 10^{-1})$ in round 1 and $(A, B, g, k) = (-2.39\times10^{-7}, 1.66 \times 10^{-3}, 1.77 \times 10^{-2}, 3.44 \times 10^{-1})$ in round 2. 

\begin{figure}[!ht]
    \centering
    \includegraphics[width=0.6\textwidth]{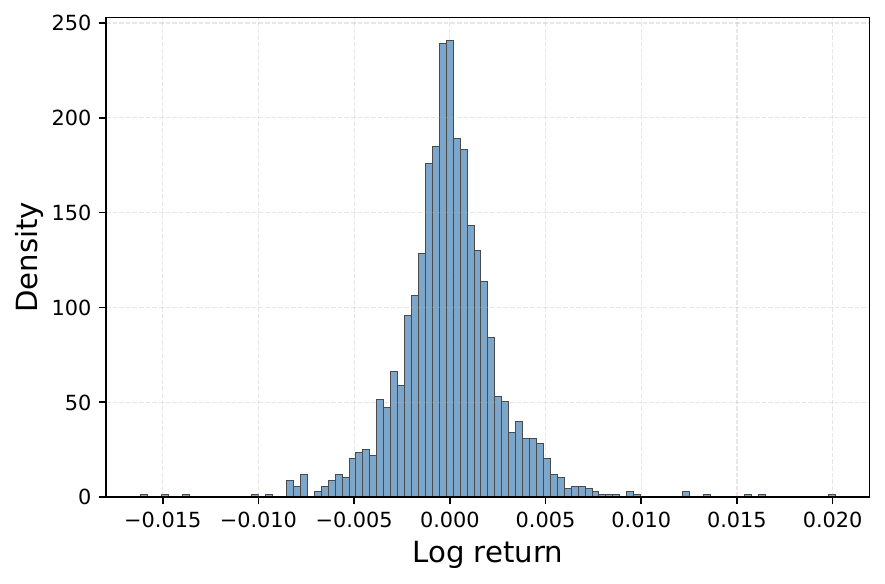}
    \caption{Histogram of the log-return data}
    \label{fig:hist_logr}
\end{figure}

\subsection{Runtime comparison}
We summarize the simulation cost and runtime of all methods and examples considered in this paper in \Cref{tab:simu_cost,tab:runtime}. In \Cref{tab:simu_cost}, each entry reports the total number of $X$ generated for the corresponding setting. Across all examples, our method and NLE have the same simulator calls and fewer than the SL method. For computation, our method and NLE use a single NVIDIA H100 NVL GPU to train the neural networks, while SL is run using a single Intel Xeon Platinum 8568Y+ CPU (96 cores). Each entry in \Cref{tab:runtime} is the average runtime over 5 repeated runs. We observe that our method and NLE have comparable runtime overall. SL is much faster than our method or NLE in most examples, since it does not require neural network training. However, its runtime depends critically on the cost of data generation. This explains why its run time exceeds the other two methods in the volatility estimation example, where the simulator is expensive and SL requires more simulator calls.

\begin{table}[!ht]
\centering
\caption{Simulation cost (number of $X$ generated) across methods and examples.}
\label{tab:simu_cost}
\begin{tabular}{l|c|c|c|c}
\hline
 & Toy example & Queuing model & g-and-k model & Volatility Estimation \\
\hline
Ours & $1.212 \times 10^6$ & $1.224 \times 10^6$ & $4.2 \times 10^7$ & $2.2 \times 10^8$ \\
NLE  & $1.212 \times 10^6$ & $1.224 \times 10^6$ & $4.2 \times 10^7$ & $2.2 \times 10^8$ \\
SL   & $5 \times 10^7$ & $2.5 \times 10^8$ & $5 \times 10^8$ & $3 \times 10^9$ \\
\hline
\end{tabular}
\end{table}

\begin{table}[!ht]
\centering
\caption{Runtime comparison (in seconds) across methods and examples.}
\label{tab:runtime}
\begin{tabular}{l|c|c|c|c}
\hline
 & Toy example & Queuing model & g-and-k model & Volatility Estimation \\
\hline
Ours & 598.27 & 2208.62 & 8884.38 & 23481.73 \\
NLE  & 929.24 & 2037.23 & 5041.24 & 25125.29 \\
SL   & 28.62 & 89.27 & 145.13 & 28291.34 \\
\hline
\end{tabular}
\end{table}

\end{document}